%% file: main.tex
\documentclass[a4paper,USenglish,autoref,numberwithinsect]{lipics-v2019}

\usepackage{ifthen}
\newboolean{long}
\setboolean{long}{true}

\usepackage{algorithm}
\usepackage[noend]{algpseudocode}
\usepackage{subfiles}

\DeclareMathOperator*{\argmin}{arg\,min}

\makeatletter
\newcommand{\algmargin}{\the\ALG@thistlm}
\makeatother
\newlength{\whilewidth}
\algdef{SE}[parWHILE]{parWhile}{EndparWhile}[1]
  {\parbox[t]{\dimexpr\linewidth-\algmargin}{%
     \hangindent\whilewidth\strut\algorithmicwhile\ #1\ \algorithmicdo\strut}}{\algorithmicend\ \algorithmicwhile}%
\algnewcommand{\parState}[1]{\State%
  \parbox[t]{\dimexpr\linewidth-\algmargin}{\strut #1\strut}}

\title{Euclidean TSP, Motorcycle Graphs, and Other New Applications of Nearest-Neighbor Chains}

\titlerunning{New Applications of Nearest-Neighbor Chains}

\author{Nil Mamano\footnote{Corresponding author.}}{Department of Computer Science, University of California Irvine, US}{nmamano@uci.edu}{https://orcid.org/0000-0003-0414-2885
}{}

\author{Alon Efrat}{Departnment of Computer Science, University of Arizona, US}{alon@cs.arizona.edu}{}{}

\author{David Eppstein}{Department of Computer Science, University of California Irvine, US}{eppstein@uci.edu}{}{Supported in part by NSF grants CCF-1618301 and CCF-1616248.}

\author{Daniel Frishberg}{Department of Computer Science, University of California Irvine, US}{dfrishbe@uci.edu}{https://orcid.org/0000-0002-1861-5439}{}

\author{Michael T. Goodrich}{Department of Computer Science, University of California Irvine, US}{goodrich@uci.edu}{https://orcid.org/0000-0002-8943-191X}{Supported in part by NSF grant 1815073.}

\author{Stephen Kobourov}{Departnment of Computer Science, University of Arizona, US}{kobourov@cs.arizona.edu}
{https://orcid.org/0000-0002-0477- 2724}{Supported in part by NSF grants CCF-1740858, CCF-1712119, DMS-1839274, and DMS-1839307.}

\author{Pedro Matias}{Department of Computer Science, University of California Irvine, US}{pmatias@uci.edu}{https://orcid.org/0000-0003-0664-9145}{}

\author{Valentin Polishchuk}{Communications and Transport Systems, ITN, Linköping University, Sweden}{valentin.polishchuk@liu.se}{}{Supported in part by project 2018-04001 from the Swedish Research Council.}

\authorrunning{N. Mamano et al.}

\Copyright{Alon Efrat, David Eppstein, Daniel Frishberg, Michael Goodrich, Stephen Kobourov, Nil Mamano, Pedro Matias, Valentin Polishchuk}



\ifthenelse{\boolean{long}}{}{
\relatedversion{A full version of the paper is available on \url{arXiv.org} under the same name (the URL will be updated for the final version).}
}

\nolinenumbers 

\ifthenelse{\boolean{long}}{
\hideLIPIcs 
}{} 

\EventEditors{John Q. Open and Joan R. Access}
\EventNoEds{2}
\EventLongTitle{42nd Conference on Very Important Topics (CVIT 2016)}
\EventShortTitle{CVIT 2016}
\EventAcronym{CVIT}
\EventYear{2016}
\EventDate{December 24--27, 2016}
\EventLocation{Little Whinging, United Kingdom}
\EventLogo{}
\SeriesVolume{42}
\ArticleNo{23}
\newcommand{\R}{\mathbb{R}}
\newcommand{\eps}{\varepsilon}

\begin{document}

\maketitle

\begin{abstract}
We show new applications of the nearest-neighbor chain algorithm, a technique that originated in agglomerative hierarchical clustering. We apply it to a diverse class of geometric problems: we construct the greedy multi-fragment tour for Euclidean TSP in $O(n\log n)$ time in any fixed dimension and for Steiner TSP in planar graphs in $O(n\sqrt{n}\log n)$ time; we compute motorcycle graphs (which are a central part in straight skeleton algorithms) in $O(n^{4/3+\eps})$ time for any $\eps>0$; we introduce a narcissistic variant of the $k$-attribute stable matching model, and solve it in $O(n^{2-4/(k(1+\eps)+2)})$ time; we give a linear-time $2$-approximation for a 1D geometric set cover problem with applications to radio station placement.
\end{abstract}

\section{Introduction}\label{sec:intro}
The \textit{nearest-neighbor chain} (NNC) technique is used for agglomerative hierarchical clustering, and has only seen one other use besides it. In this paper, we apply it to an assortment of new problems: multi-fragment TSP, straight skeletons, narcissistic $k$-attribute stable matching, and a server cover problem. These problems share a property with agglomerative hierarchical clustering, which we call global-local equivalence, and which is the key to using the NNC algorithm.  First, we review the NNC algorithm in the context of clustering.

\subsection{Prior work: NNC in hierarchical clustering}
Given a set of points, the \textit{agglomerative hierarchical clustering} problem is defined procedurally as follows: each point starts as a base cluster, and the two closest clusters are repeatedly merged until there is only one cluster left. This creates a \textit{hierarchy}, where any two clusters are either nested or disjoint. A key component of hierarchical clustering is the function used to measure distances between clusters. Popular metrics include minimum distance (or single-linkage), maximum distance (or complete-linkage), and centroid distance.

We call two clusters \textit{mutually nearest neighbors} (MNNs) if they are the nearest neighbor of each other. 
Consider this alternative, non-deterministic procedure: instead of repeatedly merging the two overall closest clusters, merge any pair of MNNs. Clearly, this may merge clusters in a different order. Nonetheless, if the cluster-distance metric satisfies a property called \textit{reducibility}, this procedure results in the same hierarchy~\cite{Bruynooghe77,bruynooghe1978,Muellner2011}. A cluster-distance metric $d(\cdot,\cdot)$ is reducible if for any clusters $A,B,C$: if $A$ and $B$ are MNNs, then 
\begin{equation}\label{eq:hcredu}
d(A\cup B, C)\geq \min{(d(A,C),d(B,C))}.
\end{equation}
In words, the new cluster $A\cup B$ resulting from merging $A$ and $B$ is \textit{not} closer to other clusters than both $A$ and $B$ were. The relevance of this property is that, if, say, $C$ and $D$ are MNNs, merging $A$ and $B$ does not break that relationship. 
The net effect is that MNNs can be merged in any order and produce the same result. Many commonly used metrics are reducible, including minimum--, maximum--, and average--distance, but others such as centroid and median distance are not.

The NNC algorithm exploits this reducibility property, which was originally observed by Bruynooghe~\cite{bruynooghe1978}. We briefly review the algorithm for hierarchical clustering, since we discuss it in detail later in the context of the new problems. For extra background on NNC for hierarchical clustering, see~\cite{murtagh1983,Muellner2011}. 
The basic idea is to maintain a stack (called \textit{chain}) of clusters. The first cluster is arbitrary. The chain is always extended with the nearest neighbor (NN) of the current cluster at the top of the chain. Note that the distance between clusters in the chain keeps decreasing, so (with an appropriate tie breaking rule) no repeated clusters or ``cycles'' occur, and the chain inevitably reaches a pair of MNNs. At this point, the MNNs are merged and removed from the chain. Crucially, after a merge happens, the rest of the chain is not discarded. Due to reducibility, every cluster in the chain still points to its NN, so the chain is still valid. The process continues from the new top of the chain.

The algorithm is efficient because each cluster is added to the chain only once, since it stays there until it is merged with another cluster. As we will see in detail for other problems, this bounds the number of iterations to be linear on the input size, with the cost of each iteration dominated by a NN computation. 

\subsection{Our contributions}\label{sec:prior}
Our key observation is that this equivalence between merging closest pairs and MNNs is not unique to hierarchical clustering. The problems in this paper, even though they are not about clustering, exhibit an analogous phenomenon, for which we coin the term \textit{global-local equivalence}. The main thesis of this paper is that NNC is an efficient algorithm for problems with global-local equivalence, which includes many more problems than hierarchical clustering.

Recently, the NNC algorithm was used for the first time outside of the domain of hierarchical clustering~\cite{eppstein2017_2,eppstein2017}. It was used in a stable matching problem where the two sets to be matched are point sets in a metric space, and each agent in one set ranks the agents in the other set by distance, with closer points being preferred. In this setting, there is a form of global-local equivalence: the stable matching is unique, and it can be obtained in two ways: by repeatedly matching the closest pair (from different sets), or by repeatedly matching MNNs. They used the NNC algorithm to solve the problem efficiently. 

In this paper, we consider global-local equivalence in the context of the new problems, and give NNC-type algorithms for them. We summarize the computational results here. See each section for extended background on the corresponding problems.

\subparagraph*{Multi-fragment TSP.}
A classic heuristic for the Euclidean Traveling Salesman Problem is the multi-fragment algorithm. While not having strong approximation guarantees, experimental results show that it performs better than other heuristics, particularly in geometric instances~\cite{Krari17,JohnMcGe97,misev11,Moscato1994AnAO,Bentley1990,bentley92}. We do not know of any subquadratic algorithm to compute the tour produced by this heuristic, which we call the multi-fragment tour. We give a $O(n\log n)$-time algorithm for computing the multi-fragment tour of a point set in any fixed dimension and using any $L_p$ metric. We also consider the Steiner TSP problem in a graph-theoretical framework~\cite{Cornuejols1985}, where we give a $O(n\sqrt{n}+k\sqrt{n}\log n)$-time algorithm for finding the multi-fragment tour through a subset of $k$ nodes in planar graphs and, more generally, graph families with $O(\sqrt{n})$-size separators. 

\subparagraph*{Straight skeletons and motorcycle graphs.}
The fastest algorithms for computing straight skeletons consist of two phases, neither of which dominates the other~\cite{Cheng2016}. The first phase is a motorcycle graph computation. The best currently known algorithm for motorcycle graphs runs in $O(P(n)+n(T(n)+\log n)\log n)$ time, where $P(n)$ and $T(n)$ are the preprocessing time and operation time (maximum between query and update) of a dynamic ray-shooting data structure for curtains in $\R^3$~\cite{Vigneron2014}. We improve this to $O(P(n)+nT(n))$. Using the structure from~\cite{agarwal93}, both algorithms run in $O(n^{4/3+\eps})$ for any $\eps>0$, but if both use the same $\eps$ in the data structure, ours is faster by a $O(\log n)$ factor.

\subparagraph*{Narcissistic $k$-attribute stable matching.}
Given that $O(n^2)$ is optimal for general stable matching instances, it is interesting to study restricted models. We introduce a narcissistic variant of the $k$-attribute model~\cite{bhatnagar2008} and give a subquadratic, $O(n^{2-4/(k(1+\eps)+2)})$-time algorithm for it, for any $\eps>0$.

\subparagraph*{Server cover.} 
We give a linear-time $2$-approximation for a one-dimensional version of a \textit{server coverage} problem: given the locations of $n$ clients and $m$ servers, which can be seen as houses and telecommunication towers, the goal is to assign a ``signal strength'' to each communication tower so that they reach all the houses, minimizing the cost of transmitting the signals. This improves upon the $O(m + n\log m)$-time algorithm by Alt et al.~\cite{carrots} with the same approximation ratio.

\paragraph*{Paper organization.}
Section~\ref{sec:snn} introduces a new data structure, which we call the \textit{soft nearest-neighbor data structure}. 
Section~\ref{sec:mfh} solves multi-fragment Euclidean TSP with a variant of NNC that uses this structure.
\ifthenelse{\boolean{long}}{Sections~\ref{sec:motorcycle},~\ref{sec:kattr}, and~\ref{sec:clientcover} are on motorcycle graphs, narcissistic $k$-attribute stable matching, and server cover, respectively.

Thus, Section~\ref{sec:mfh} relies on Section~\ref{sec:snn}, but the other sections are independent, self-contained, and in no particular order.}{We extend this result to Steiner TSP in Appendix~\ref{sec:mfgraph}.
Section~\ref{sec:motorcycle} is on motorcycle graphs.
Due to space constraints, we defer narcissistic $k$-attribute stable matching and server cover, respectively, to appendices~\ref{sec:kattr} and~\ref{sec:clientcover}, which are independent, self-contained, and in no particular order.}
We give concluding remarks in Section~\ref{sec:conclusions}.

\section{The Soft Nearest-Neighbor Data Structure}\label{sec:snn}
Throughout this section, we consider points in $\R^\delta$, for some fixed dimension $\delta$, and distances measured under any $L_p$ metric $d(\cdot,\cdot)$.
We begin with a formal definition of the structure and the main result of this section. 
\begin{definition}[Dynamic soft nearest-neighbor data structure]
Maintain a dynamic set of points, $P$, subject to insertions, deletions, and \textit{soft nearest-neighbor} queries: given a query point $q$, return either of the following:
\begin{itemize}
    \item The nearest neighbor of $q$ in $P$: $p^*=\argmin_{p\in P} d(q,p)$.
    \item A pair of points $p,p'$ in $P$ satisfying $d(p,p')<d(q,p^*)$. 
\end{itemize}
\end{definition}

\begin{theorem}\label{thm:snn}
In any fixed dimension, and for any $L_p$ metric, there is a dynamic soft nearest-neighbor data structure that maintains a set of $n$ points with $O(n\log n)$ preprocessing time and $O(\log n)$ time per operation (queries and updates).
\end{theorem}

We label the two types of answers to soft nearest-neighbor (SNN) queries as \emph{hard} or \emph{soft}. 
A ``standard'' NN data structure is a special case of a SNN structure that always gives hard answers. However, in light of Theorem~\ref{thm:snn}, a standard NN structure would not be as efficient as a SNN structure. For comparison, the best dynamic NN structure in $\R^2$ requires $O(\log^5 n)$ time per operation~\cite{Chan2010,KapMulRod-16}.

In our implementation, we use the following data structure. Given a point set $P$ and a point $q$, let $p^*_i$ denote the $i$-th closest point to $q$ in $P$.

\begin{definition}[Dynamic $\eps$-approximate $k$ nearest-neighbor ($k$-ANN) data structure]
Maintain a dynamic set of points, $P$, subject to insertions, deletions, and \textit{$\eps$-approximate $k$ nearest-neighbor} queries: given a query point $q$ and an integer $k$ with $1\leq k\leq |P|$, return $k$ points $p_1,\ldots,p_k\in P$ such that, for each $p_i$, $d(q,p_i)\leq (1+\eps)d(q,p^*_i)$, where $\eps>0$ is a constant known at construction time\footnote{Some approximate nearest-neighbor data structures~\cite{arya1998optimal} do not need to know $\eps$ at construction time, and, in fact, allow $\eps$ to be part of the query and to be different for each query. Clearly, such data structures are also valid for our needs.}. 
\end{definition}

We reduce each SNN query to a single $k$-ANN query with constant $\eps$ and $k$. Once we show this reduction, Theorem~\ref{thm:snn} will follow from the following result by Arya et al.~\cite{arya1998optimal}:
\begin{lemma}[\cite{arya1998optimal}]
In any fixed dimension, and for any $L_p$ metric, there is a dynamic $\eps$-approximate $k$ nearest-neighbor data structure with $O(n\log n)$ preprocessing time and $O(\log n)$ time per operation (query and updates) for constant $k$ and $\eps>0$.
\end{lemma}

\subsection{Soft nearest-neighbor implementation}
We maintain the point set $P$ in a dynamic $k$-ANN structure ($\eps$ depends on the metric space, and will be determined later).
In what follows, $q$ denotes an arbitrary query point and $p^*_i$ the $i$-th closest point to $q$ in $P$. For ease of presentation, we assume throughout the section that $d(q,p_1^*)=1$. This scaling does not affect any result. Queries rely on the following lemma.

\begin{lemma}\label{lem:pack}
Consider a query $(q,k)$ to a $k$-ANN structure. If none of the $k$ returned points, $p_1,\ldots,p_k,$ is $p_1^*$, then, for each $p_i$ with $1\leq i\leq k$, we have that $d(q,p_i)\leq (1+\eps)^i$.
\end{lemma}

\begin{proof}
For $i=1$, the fact follows immediately from the definition of the $k$-ANN structure (and the assumption that $d(q,p_1^*)=1$). For $i=2,\ldots,k$, note that 
$d(q,p_i^*)\leq d(q, p_{i-1})$. This is because there are at least $i$ points within distance $d(q, p_{i-1})$ of $q$: $p_1^*,p_1,\ldots,p_{i-1}$. Thus,
$d(q, p_i)\leq (1+\eps)d(q, p_i^*)\leq(1+\eps)d(q, p_{i-1}).$
The claim follows by induction.
\ifthenelse{\boolean{long}}{
It is illustrated in Figure~\ref{fig:snn}.
\begin{figure}
    \centering
    \includegraphics[width=0.45\textwidth]{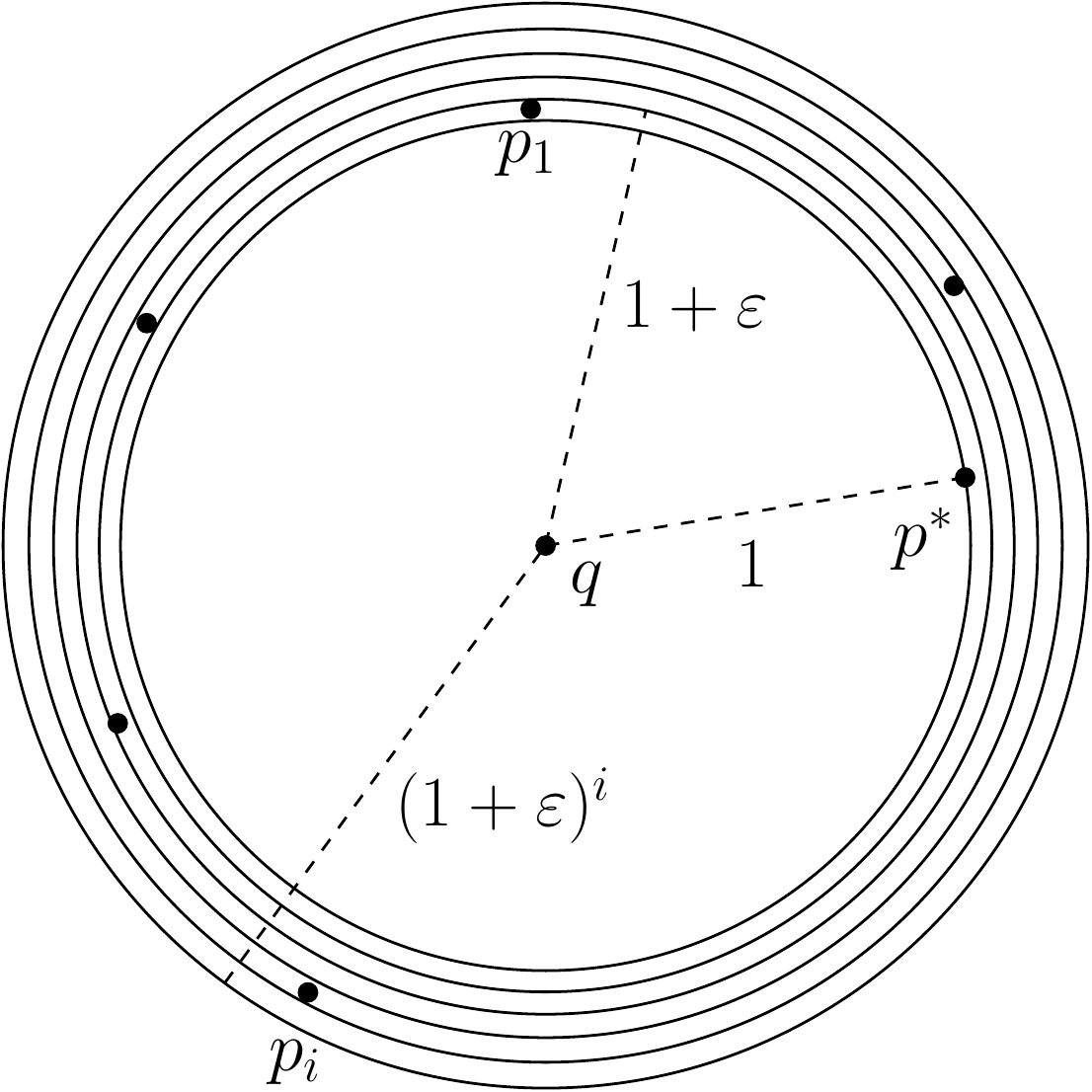}
    \caption{Setting of Lemma~\ref{lem:pack}. The circles centered at $q$ have radius $1,1+\varepsilon,(1+\varepsilon)^2,\ldots$. The point $p^*$ is the true NN of $q$, and the other points are the first points returned by the $k$-ANN structure. If $p^*$ is not one of the returned points, each point $p_i$ is within distance $(1+\varepsilon)^i$ of $q$.}
    \label{fig:snn}
\end{figure}
}{}
\end{proof}

Let $S(q,r_1,r_2)$ denote a closed shell centered at $q$ with inner radius $r_1$ and outer radius $r_2$ (i.e., $S(q,r_1,r_2)$ is the difference between two balls centered at $q$, the bigger one of radius $r_2$ and the smaller one of radius $r_1$). From Lemma~\ref{lem:pack}, we get the following.

\begin{corollary}\label{cor:pack}
Consider a query $(q,k)$ to an $k$-ANN structure. If none of the $k$ returned points, $p_1,\ldots,p_k$, is $p_1^*$, then they all lie in $S(q,1,(1+\eps)^k)$.
\end{corollary}

We call a pair $(\eps,k)$ \textit{valid parameters} if, in any set of $k$ points inside a shell with inner radius $1$ and outer radius $(1+\eps)^{k}$, there must exist two points $p,p'$ satisfying $d(p,p')<1$. Suppose that $(\eps^*,k^*)$ are valid parameters. Initially, we construct the $k$-ANN structure using $1+\eps^*$ as the approximation factor. Then we answer queries as in Algorithm~\ref{alg:snn}.

\begin{algorithm}
\caption{Soft nearest-neighbor query.}
\label{alg:snn}
\begin{algorithmic}
\State Ask query $(q, k^*)$ to the $k$-ANN structure initialized with $\eps^*$.
\State Measure the distance between each pair of the $k^*$ returned points, $p_1,\ldots,p_{k^*}$.
\If{any pair $(p,p')$ satisfy $d(p,p')<1$}
\State\Return $p,p'$.
\Else 
\State\Return the point $p_i$ that is closest to $q$.
\EndIf
\end{algorithmic}
\end{algorithm}

\begin{lemma}
If $(\eps^*,k^*)$ are valid parameters, Algorithm~\ref{alg:snn} is correct.
\end{lemma}

\begin{proof}
If a pair $p,p'$ of points returned by the $k$-ANN structure satisfy $d(p,p')<1$, $p$ and $p'$ are a valid soft answer to the SNN query. Thus, consider the alternative case: no pair of the $k^*$ returned points is at distance ${<1}$. 
Then, because ($\eps^*,k^*$) are valid, at least one of the returned points must be outside of $S(q,1,(1+\eps^*)^{k^*})$. By the contrapositive of Corollary~\ref{cor:pack}, one of them must be $p_1^*$.
\end{proof}

\ifthenelse{\boolean{long}}{
As a side note, a SNN structure always returns a hard answer when queried from a point that is part of the closest pair of the set of points it maintains, as there is no \textit{closer} pair. In this way, a SNN structure can be used to find the closest pair in $(\R^\delta,L_p)$, for constant $\delta$, in $O(n\log n)$ time by querying from every point. This matches the known runtimes in the literature~\cite{Bespamyatnikh1998}.
}{}

\ifthenelse{\boolean{long}}{
\subsection{Choice of parameters}\label{sec:params}

\subfile{sectionparameters}

}{
We left open the question of finding valid parameters $(\eps^*,k^*)$, which we defer to Appendix~\ref{sec:params}. In particular, it is not hard to see that they exist in any metric space $(\R^\delta,L_p)$ (see Lemma~\ref{lem:kissing}).
}

\section{Multi-Fragment Euclidean TSP}\label{sec:mfh}
The \textit{Euclidean Travelling Salesperson Problem} asks to find, given a set of points, a closed \textit{tour} (a closed polygonal chain) through all the points of shortest length. The problem is NP-hard even in this geometric setting, but a polynomial-time approximation scheme is known~\cite{arora1998polynomial}.

In this section, we consider a classic greedy heuristic for constructing TSP tours, \textit{multi-fragment} TSP. In this algorithm, each point starts as a single-node path. While there is more than one path, connect the two closest paths. Here, the distance $d(a,b)$ between two paths $a,b$ is measured as the minimum distance between their endpoints, and connecting two paths means adding the edge between their closest endpoints. Once there is a single path left, connect their endpoints. We call the tour resulting from this process the \textit{multi-fragment tour}.

The multi-fragment algorithm was proposed by Bentley~\cite{bentley92} specifically in the geometric setting. Its approximation ratio is $O(\log n)$~\cite{ONG1984273,Brecklinghaus15}. Nonetheless, it is used in practice due to its simplicity and empirical support that it  generally performs better than other heuristics~\cite{Krari17,JohnMcGe97,misev11,Moscato1994AnAO,Bentley1990}.

We are interested in the complexity of computing the multi-fragment tour.
A straightforward implementation of the multi-fragment algorithm is similar to Kruskal's minimum spanning tree algorithm: sort the $\binom{n}{2}$ pairs of points by increasing distances and process them in order: for each pair, if the two points are endpoints of separate paths, connect them. The runtime of this algorithm is $O(n^2\log n)$.
Eppstein~\cite{eppstein2000fast} uses dynamic closest pair data structures to compute the multi-fragment tour in $O(n^2)$ time (for arbitrary distance matrices). Bentley~\cite{bentley92} gives a $K$-$d$ tree-based implementation and says that it appears to run in $O(n\log n)$ time on uniformly distributed points in the plane.
We give a NNC-type algorithm that compute the multi-fragment tour in $O(n\log n)$ in any fixed dimensions. We do not know of any prior worst-case subquadratic algorithm.

\subsection{Global-local equivalence in multi-fragment TSP}
Since the multi-fragment algorithm operates on paths rather than points, it will be convenient to think of the input as a set of paths (a path is an open polygonal chain, although, in the context of the algorithm, only the coordinates of the endpoints are relevant). The input to Euclidean TSP corresponds to a set of paths where all paths are single-point paths.
Consider the following two strategies for constructing a tour from a set of paths, where we use $a\cup b$ to denote the path resulting from connecting paths $a$ and $b$:
\begin{itemize}
    \item While there is more than one path, connect two paths using one of the following strategies:
    \begin{enumerate}
        \item Connect the closest pair of paths.\label{strat:cp}
        \item Connect two mutually nearest-neighbor paths.\label{strat:mnn}
    \end{enumerate}
    \item Connect the two endpoints of the final path.
\end{itemize}

Strategy~\ref{strat:cp} corresponds to the multi-fragment algorithm. Note that Strategy~\ref{strat:mnn} is non-deterministic, and that Strategy~\ref{strat:cp} is a special case of Strategy~\ref{strat:mnn}. In this section, we show that \textit{any} execution of Strategy~\ref{strat:mnn} computes the multi-fragment tour.
Note the similarity between multi-fragment TSP and hierarchical clustering. We can see that in multi-fragment TSP we have a notion equivalent to reducibility in agglomerative hierarchical clustering (Equation~\ref{eq:hcredu}).

\begin{lemma}[Reducibility in multi-fragment TSP]\label{lem:reducibility}
Let $a,b,$ and $c$ be paths. Then, $d(a\cup b, c)\geq \min{(d(a,c),d(b,c))}$.
\end{lemma}

\begin{proof}
The distance between paths is defined as the minimum distance between their endpoints, and
the two endpoints of $a\cup b$ are a subset of the four endpoints of $a$ and $b$. 
\end{proof}

\begin{lemma}[Global-local equivalence in multi-fragment TSP]\label{lem:mfgle}
Assuming that there are no ties in the pairwise distances between paths, strategies~\ref{strat:cp} and~\ref{strat:mnn} produce the same tour.
\end{lemma}

\newcommand{\proofmfgle}{We adapt the proof of global-local equivalence for agglomerative hierarchical clustering presented in~\cite{Muellner2011}. We note that we can break ties using a consistent rule, such as breaking ties by the smallest index in the input.

\begin{proof}
Let $P$ be a set of paths, and let $S_1(P)$ denote the sequence of path pairs connected by Strategy~\ref{strat:cp} starting from $P$, and $T_1(P)$ the corresponding resulting tour.
Similarly, let $S_2(P)$ denote one of the possible sequences of path pairs connected by an instantiation of Strategy~\ref{strat:mnn} starting from $P$, and $T_2(P)$ the corresponding tour.
We need to show that $T_1(P) = T_2(P)$.

Proceed by induction on $|P|$. If $|P|=1$, both tours are the same because no connections happen. Thus, let $|P|>1$.
Let $(a,b)$ be the first pair of paths in $S_2(P)$. Then, consider the set $P'=P\setminus\{a,b\}\cup\{a\cup b\}$. The tour $T_2(P)$ can be seen as the tour obtained by starting from the set $P'$ and connecting the same paths as in $S_2(P)$ after the first connection $(a,b)$. Note that $|P|=|P|-1$. Thus, by the inductive hypothesis, $T_2(P)=T_1(P')$.
The bulk of the proof is to show that $T_1(P')=T_1(P)$.

First, note that $(a,b)$ is in $S_1(P)$: initially, $a$ and $b$ are MNN paths (since they are the first pair chosen by Strategy~\ref{strat:mnn}). Then, they remain so throughout the algorithm until they are connected. This is because \textit{(i)} MNN paths are not connected with other paths, and \textit{(ii)} by reducibility (Lemma~\ref{lem:reducibility}), MNN paths stay so even if other paths are connected (i.e., if $x$ and $y$ are connected, $x\cup y$ is not closer to $a$ (or $b$) than the closest of $x$ and $y$).

Let $(a,b)$ be the $k$-th pair in $S_1(P)$. Next, we show that the first $k-1$ pairs in $S_1(P)$ and $S_1(P')$ are the same and in the same order.
Let $(x,y)$ be the first pair of paths in $S_1(P)$.
By Strategy~\ref{strat:cp}, $d(x,y)$ is minimum among all distances between paths in $P$.
By Lemma~\ref{lem:reducibility}, $a\cup b$ is not closer to $x$ or $y$ than $a$ or $b$. Thus, in $S_1(P')$, $d(x,y)$ is also minimum, so $(x,y)$ is also the first element in $S_1(P')$. The claim for the next $k-2$ pairs follows analogously by  induction.

Finally, note that after the first $k$ connections in $S_1(P)$ and the first $k-1$ connections in $S_1(P')$, the corresponding partial solutions are the same. After that point, all the connections, and the final solution, must be the same in both, so $T_1(P)=T_1(P')$.
\end{proof}
}

\ifthenelse{\boolean{long}}{\proofmfgle
}{This lemma follows from an argument similar to the one presented in~\cite{Muellner2011} for the analogous result for agglomerative hierarchical clustering. We show the details of the proof in Appendix~\ref{app:mfgle}.
}

We note that Lemma~\ref{lem:mfgle} holds for arbitrary distance matrices.

\subsection{Soft nearest-neighbor chain for multi-fragment Euclidean TSP}

Given that we have global-local equivalence (Lemma~\ref{lem:mfgle}), we can use the NNC algorithm to compute the multi-fragment tour using Strategy~\ref{strat:mnn}. A straightforward adaptation of the NNC algorithm, paired with the NN structure from~\cite{Chan2010,KapMulRod-16}, yields a $O(n\log^5 n)$ runtime for $\R^2$. 
However, we skip this result and jump directly to our main result:

\begin{theorem}\label{thm:snncmf}
The multi-fragment tour of a set of $n$ points in any fixed dimension, and under any $L_p$ metric, can be computed in $O(n\log n)$ time.
\end{theorem}

We use a variation of the NNC algorithm that uses a SNN structure instead of the usual NN structure, which we call \textit{soft nearest-neighbor chain} (SNNC). For this, we need a SNN structure for paths instead of points.
That is, a structure that maintains a set of (possibly single-node) paths, and, given a query path $q$, returns the closest path to $q$ or two paths which are closer. 

\ifthenelse{\boolean{long}}{
\subparagraph*{A soft nearest-neighbor structure for paths.}
\subfile{sectionsnnforpaths}

}{
Appendix~\ref{app:snnpath} shows how to adapt the SNN structure for points into a SNN structure for paths. In short, we achieve this by storing the set of path endpoints in a SNN structure for points, and, given a query path, we do a query for each endpoints. Appendix~\ref{app:snnpath} shows how to handle the technicality that a soft answer could return two endpoints which belong to the same path. We solve this by setting the parameters $(\eps,k)$ of the data structure so that, in soft answers, we get three pairwise closer points instead. 
}

\subparagraph*{The soft nearest-neighbor chain algorithm.}

We use a SNN for paths.
In the context of this algorithm, let us think of a SNN answer, hard or soft, as being a set of two paths. If the answer is hard, then one of the paths returned in the answer is the query path itself, and the remaining path is its NN. Now, we can establish a comparison relationship between SNN answers (independently of their type):
given two SNN answers $\{a,b\}$ and $\{c,d\}$, we say that $\{a,b\}$ is \emph{better} than $\{c,d\}$ if and only if $d(a,b) < d(c,d)$.

The input is a set of paths, where we again assume unique distances.
The algorithm maintains a stack (the chain) of \textit{nodes}, where each node consists of a pair of paths (with the exception of the first node in the chain, which contains a single path). In particular, each node in the chain is the best SNN answer among two queries for the two paths in the predecessor node (when querying from a path, we remove it from the structure temporarily, so that the answer is not itself).

The algorithm starts with an arbitrary path in the chain. If the chain ever becomes empty and there is still more than one path, the chain is restarted at an arbitrary path. If the best answer we get from the SNN structure is precisely the node currently at the top of the chain, we connect both paths contained in it and remove the node and its predecessor from the chain. Otherwise, we append the answer to the top of the chain as a new node. See \autoref{alg:snccpath} for a full description of the algorithm and Figure~\ref{fig:snncmftsp} for a snapshot of the algorithm.

\newcommand{\figsnncmftsp}{
\begin{figure}
    \centering
    \includegraphics[width=0.99\textwidth]{snncmftsp}
    \caption{\textbf{Left:} a set of paths (some of which are single points) and a possible chain, where the nodes are denoted by dashed lines and appear in the chain according to the numbering. Note that the first node is the only one containing a single path, and that all the nodes in the chain are soft answers except the fourth node. \textbf{Right:} Nearest-neighbor graph of the set of paths. For each path, a dashed/red arrow points to its NN. Further, the arrows start and end at the endpoints determining the minimum distance between the paths.}
    \label{fig:snncmftsp}
\end{figure}
}

\ifthenelse{\boolean{long}}{
\figsnncmftsp
}{}

\begin{algorithm}
\caption{Soft nearest-neighbor chain algorithm for multi-fragment Euclidean TSP}
\label{alg:snccpath}
\begin{algorithmic}
\State Initialize an empty stack (the chain).
\State Initialize a one-node path for every input point.
\State Initialize a SNN structure $S$ for paths (as in Lemma~\ref{lem:snnpath}) with the set of one-node paths.
\While{there is more than one path in $S$}
    \If{the chain is empty}
    \State add an arbitrary path from $S$ to it.
    \Else
    \State Let $U=\{u,v\}$ be the node at the top of the chain (or just $u$ for the first node).
    \State Remove $u$ from $S$, query $S$ with $u$, and re-add $u$ to $S$. 
    \State Remove $v$ from $S$, query $S$ with $v$, and re-add $v$ to $S$.
    \State Let $A$ be the best answer.
    \If{$A = U$}
        \parState{Connect $u$ and $v$, remove them from $S$, add $u\cup v$ to $S$, and remove $U$ and its predecessor from the chain.}
    \Else
        \State Add $A$ to the chain.
    \EndIf
    \EndIf
\EndWhile
\State Connect the two endpoints of the remaining path in $S$.
\end{algorithmic}
\end{algorithm}

\begin{lemma}
\label{lem:invariants}
The following invariants hold at the beginning of each iteration of \autoref{alg:snccpath}:
\begin{enumerate}
    \item The SNN structure $S$ contains a set of disjoint paths.
    \item If node $R$ appears after node $S$ in the chain, then $R$ is better than $S$. \label{inv:dec_distance}
    \item Every path in $S$ appears in at most two nodes in the chain, in which case they consist of occurrences in two consecutive nodes. \label{inv:no_repetitions}
    \item The chain only contains paths in $S$.\label{inv:nodes_unmatched}
\end{enumerate}
\end{lemma}

\begin{proof}

~

\begin{enumerate}

    \item The claim holds initially. Each time two paths $a,b$ are connected, one endpoint of each becomes an internal point in the new path $a\cup b$. Since $a$ and $b$ are removed from $S$, no path can be connected to those endpoints. 

    \item We show it for the specific case where $R$ is immediately after $S=\{s,t\}$ in the chain, which suffices. Note that $R$ is different than $S$, or it would not have been added to the chain. We distinguish between two cases:
    \begin{itemize}
        \item $s$ and $t$ were MNNs when $R$ was added. Then, $R$ had to be a soft answer from $s$ or $t$, which would have to be better than $\{s,t\}$.
        \item $s$ and $t$ were not MNNs when $R$ was added. Then, $s$ had a closer path than $t$ (wlog). Thus, whether the answer for $s$ was soft or hard, the answer had to be better than $\{s,t\}$.
    \end{itemize}

    \item Assume, for a contradiction, that three nodes $X=\{p,x\}$, $Y$ and $Z=\{p,z\}$ appear in the chain in order $X$, $Y$, $Z$ (not necessarily consecutively, and with $p$ possibly in $Y$). By Invariant~\ref{inv:dec_distance}, $Z$ is better than $Y$. It is easy to see that if $z_1$ and $z_2$ are the two endpoints of path $z$, then $z_1$ and $z_2$ were endpoints of paths since the beginning of the algorithm. Thus, the answer for $p$ when $X$ was at the top of the chain had to be a pair at distance at most $\min(d(p,z_1), d(p,z_2))$. Note that $\min(d(p,z_1), d(p,z_2))=d(p,z)$, contradicting that $Z$ is better than $Y$.

    \item We show that no node in the chain contains paths that have already been connected to form bigger paths. Whenever the two paths in the node at the top of the chain are connected, we remove the node from the chain. By Invariant~\ref{inv:no_repetitions}, their only other occurrence can only be in the predecessor node, which is also removed. In addition, since the paths are removed from $S$ when they are connected, newly added nodes to the chain only contain paths that have not been connected to form bigger path yet.
    
\end{enumerate}
\end{proof}

\begin{lemma}\label{lem:snncpathcorrectness}
Paths connected in \autoref{alg:snccpath} are MNNs in the set of paths in the SNN structure.
\end{lemma}

\begin{proof}
Let $\{u,v\}$ be the node at the top of the chain, and $A$ the best SNN answer among the $u$ and $v$ queries.
If $u$ and $v$ are not MNNs, at least one of them, $u$ (wlog), has a closer path than the other, $v$, so the answer for $u$ cannot be $\{u,v\}$. By the contrapositive, if the best answer $A$ is $\{u,v\}$, then $u$ and $v$ are MNNs. In the algorithm, $u$ and $v$ are connected precisely when $A=\{u,v\}$.
\end{proof}

\begin{proof}[Proof of \autoref{thm:snncmf}]
We show that Algorithm~\ref{alg:snccpath} computes the multi-fragment tour in $O(n\log n)$ time. 
In particular, it implements Strategy~\ref{strat:mnn}: the SNN structure maintains a set of paths, and the algorithm repeatedly connects MNNs (Lemma~\ref{lem:snncpathcorrectness}) By global-local equivalence (Lemma~\ref{lem:mfgle}), this produces the multi-fragment tour.

Note that the chain is acyclic in the sense that each node contains a path from the current set of paths in $S$ (Invariant~\ref{inv:nodes_unmatched}) not found in previous nodes (by Invariant~\ref{inv:no_repetitions}). Thus, the chain cannot grow indefinitely, so, eventually, paths get connected. The main loop does not halt until there is a single path.

If there are $n$ paths at the beginning, there are $2n-1$ different paths throughout the algorithm. This is because each connection removes two paths and adds one new path. At each iteration, either two paths are connected, which happens $n-1$ times, or one node is added to the chain. Since there are $n-1$ connections, each of which triggers the removal of two nodes in the chain, the total number of nodes removed from the chain is $2n-2$. Since every node added is removed, the number of nodes added to the chain is also $2n-2$. Thus, the total number of iterations is $3n-3$. Therefore, the total running time is $O(P(n)+nT(n))$, where $P(n)$ and $T(n)$ are the preprocessing and operation time of a SNN structure for paths. By Lemma~\ref{lem:snnpath}, this can be done in $O(n\log n)$ time.
\end{proof}

\ifthenelse{\boolean{long}}{
Incidentally, the \textit{maximum-weight matching} problem has the same type of global-local equivalence: consider the classic factor-$2$ approximation greedy algorithm that picks the heaviest edge at each iteration and discards the neighboring edges (edges with a shared endpoint)~\cite{avis83}. An alternative algorithm that picks any edge which is heavier than its neighbors produces the same matching~\cite{Hoepman2004SimpleDW}, which we call the greedy matching. In the geometric setting, we are interested in matching points minimizing distances instead (but the mentioned results still hold). The greedy algorithm repeatedly matches the closest pair. It is possible to modify Algorithm~\ref{alg:snccpath} to find the greedy matching in $O(n\log n)$ time in any fixed dimension. The algorithm is, in fact, simpler, because the SNN structure only needs to maintain points instead of paths, and matched points are removed permanently (unlike connected paths which are re-added to the set of paths). However, this is not a new result, as there is a dynamic closest pair data structure with $O(\log n)$ time per operation~\cite{Bespamyatnikh1998} which can be used to find the greedy matching in the same time bound.
}{}

\ifthenelse{\boolean{long}}{
\subsection{Steiner TSP}\label{sec:mfgraph}

\subfile{sectionsteinertsp}

}{}

\section{Motorcycle Graphs}\label{sec:motorcycle}
An important concept in geometric computing is the straight skeleton~\cite{aichholzer1996novel}. It is a tree-like structure similar to the medial axis of a polygon, but which consists of straight segments only.
Given a polygon, consider a shrinking process where each edge moves inward, at the same speed, in a direction perpendicular to itself. The straight skeleton of the polygon is the trace of the vertices through this process. Some of its applications include computing offset polygons~\cite{Eppstein1999}, medical imaging~\cite{cloppet2000}, polyhedral surface reconstruction~\cite{oliva1996,Barequet2003}, and computational origami~\cite{DEMAINE20003}. It is a standard tool in geometric computing software~\cite{cacciola04}.

The current fastest algorithms for computing straight skeletons consist of two main steps~\cite{Cheng2007,Huber2011,huber12}. The first step is to construct a motorcycle graph induced by the reflex vertices of the polygon. The second step is a lower envelope computation. With current algorithms, the first step is more expensive, but it only depends on the number of reflex vertices, $r$, which might be smaller than the total number of vertices, $n$. Thus, no step dominates the other in every instance. In this section, we discuss the first step. The second step can be done in $O(n\log n)$ time for simple polygons~\cite{bowers2014faster}, in $O(n\log n\log r)$ time for arbitrary polygons~\cite{Cheng2016}, and in $O(n\log n\log m)$ time for planar straight line graphs with $m$ connected components~\cite{bowers2014faster}.

The motorcycle graph problem can be described as follows (see Figure~\ref{fig:mc}, top)~\cite{Eppstein1999}. The input consists of $n$ points in the plane, with associated directions and speeds (the motorcycles). Consider the process where all the motorcycles start moving at the same time, in their respective directions and speeds. Motorcycles leave a trace behind that acts as a ``wall'' such that other motorcycles crash and stop if they reach it. Some motorcycles escape to infinity while others crash against the traces of other motorcycles. The motorcycle graph is the set of traces.

Most existing algorithms rely on three-dimensional ray-shooting queries. Indeed, if time is seen as the third dimension, the position of a motorcycle starting to move from $(x,y)$, at speed $s$, in the direction $(u,v)$, forms a ray (if it escapes) or a segment (if it crashes) in three dimensions, starting at $(x,y,0)$ in the direction $(u,v,1/s)$. In particular, the impassable traces left behind by the motorcycles correspond to infinite vertical ``curtains'' -- wedges or trapezoidal slabs, depending on whether they are bounded below by a ray or a segment.

Thus, ray-shooting queries help determine which trace a motorcycle would reach first, if any. Of course, the complication is that as motorcycles crash, their potential traces change. Early algorithms handle this issue by computing the crashes in chronological order~\cite{Eppstein1999,Cheng2007}. The best previously known algorithm, by Vigneron and Yan~\cite{Vigneron2014}, is the first that computes the crashes in non-chronological order.
Our NNC-based algorithm improves upon it by reducing the number of ray-shooting queries needed from $O(n\log n)$ to $3n$, and simplify significantly the required data structures. It is also non-chronological, but follows a completely new approach.

\newcommand{\figmc}{
\begin{figure}[t]
    \centering
    \includegraphics[width=.78\linewidth]{motorcyclegraph}
    \caption{\textbf{Top:} an instance input with uniform velocities and its corresponding motorcycle graph. \textbf{Bottom:} snapshots of the NNC algorithm before and after determining all the motorcycles in a NN cycle found by the chain: the NN of the motorcycle at the top, $m$, is $m'$, which is already in the chain. Note that some motorcycles in the chain have as NNs motorcycles against the traces of which they do not crash in the final output. That is expected, because these motorcycles are still undetermined (e.g., as a result of clipping the curtain of $m'$, the NN of its predecessor in the chain changes).}
    \label{fig:mc}
\end{figure}
}

\ifthenelse{\boolean{long}}{
\figmc
}{}

\subsection{Algorithm description}

In the algorithm, we distinguish between \textit{undetermined} motorcycles, for which the final location is still unknown, and \textit{determined} motorcycle, for which the final location is already known. We use a dynamic three-dimensional ray-shooting data structure. In the data structure, determined motorcycles have wedges or slabs as curtains, depending on whether they escape or not. Undetermined motorcycles have wedge curtains, as if they were to escape. Thus, curtains of undetermined motorcycles may reach points that the corresponding motorcycles never get to.

For an undetermined motorcycle $m$, we define its nearest neighbor to be the motorcycle, determined or not, against which $m$ would crash next according to the set of curtains in the data structure. Motorcycles that escape may have no NN. Finding the NN of a motorcycle $m$ corresponds to one ray-shooting query. Note that $m$ may actually not crash against the trace of its NN, $m'$, if $m'$ is undetermined and happens to crash early. On the other hand, if $m'$ is determined, then $m$ definitely crashes into its trace.

We begin with all motorcycles as undetermined. Our main structure is a chain (a stack) of undetermined motorcycles such that each motorcycle is the NN of the previous one.
In contrast to typical applications of the NNC algorithm, here ``proximity'' is not symmetric: there may be no ``mutually nearest-neighbors''. In fact, the only case where two motorcycles are MNNs is the degenerate case where two motorcycles reach the same point simultaneously. That said, mutually nearest neighbors have an appropriate analogous in the asymmetric setting: \textit{nearest-neighbor cycles}, $m_1\rightarrow m_2\rightarrow\cdots\rightarrow m_k\rightarrow m_1$. Our algorithm relies on the following key observation: if we find a nearest-neighbor cycle of undetermined motorcycles, then each motorcycle in the cycle crashes into the next motorcycle's trace. This is easy to see from the definition of nearest neighbors, as it means that no motorcycle outside the cycle would ``interrupt'' the cycle by making one of them crash early. Thus, if we find such a cycle, we can determine all the motorcycles in the cycle at once (this can be seen as a type of chronological global-local equivalence).

Starting from an undetermined motorcycle, following a chain of nearest neighbors inevitably leads to \textit{(a)} a motorcycle that escapes, \textit{(b)} a motorcycle that is determined, or \textit{(c)} a nearest-neighbor cycle. In all three cases, this allows us to determine the motorcycle at the top of the chain, or, in Case (c), all the motorcycles in the cycle. See Figure~\ref{fig:mc}, bottom. Further, note that we only modify the curtain of the newly determined motorcycle(s). Thus, if we determine the motorcycle $m$ at the top of the chain, only the NN of the second-to-last motorcycle in the chain may have changed, and similarly in the case of the cycle. Consequently, the rest of the chain remains consistent. Algorithm~\ref{alg:mc} shows the full pseudo code.

\begin{algorithm}
\caption{Nearest-neighbor-chain algorithm for motorcycle graphs.}
\label{alg:mc}
\begin{algorithmic}
\State Initialize a ray-shooting data structure with the wedges for all the motorcycles (according to their status as undetermined).
\State Initialize an empty stack (the chain).
\While{there are undetermined motorcycles}
    \State If the chain is empty, add an arbitrary undetermined motorcycle to it.
    \parState{Let $m$ be the motorcycle at the top of the chain. Do a query for the NN of $m$. If there is any, denote it by $m'$. There are four cases:
        \begin{enumerate}[(a)]
            \item $m$ does not have a NN: $m$ escapes. Remove $m$ from the chain and mark it as determined (its curtain does not change).
            \item $m'$ is determined (i.e., its curtain is final): $m$ crashes into it. Clip the curtain of $m$ into a slab, mark $m$ as determined, and remove $m$ and the previous motorcycle from the chain. (We remove the second-to-last motorcycle because it had $m$ as NN, and after clipping $m$'s curtain, the previous motorcycle may have a different NN.)
            \item $m'$ is undetermined and already in the chain: then, all the motorcycles in the chain, from $m'$ up to $m$ (which is the last one) form a nearest-neighbor cycle, and they will all crash against the trace of the next motorcycle in the cycle. We clip all their curtains, mark them all as determined, and remove them and the motorcycle immediately before $m'$ from the chain.
            \item $m'$ is undetermined and not in the chain: add $m'$ to the chain.
        \end{enumerate}}
\EndWhile
\end{algorithmic}
\end{algorithm}

\subsection{Analysis}
Clearly, every motorcycle eventually becomes determined, and we have already argued in the algorithm description that irrespective of whether it becomes determined through Case (a), (b), or (c), its final position is correct. Thus, we move on to the complexity analysis. Each ``clipping'' update can be seen as an update to the ray-shooting data structure: we remove the wedge and add the slab.

\begin{theorem}
Algorithm~\ref{alg:mc} computes the motorcycle graph in time $O(P(n)+nT(n))$, where $P(n)$ and $T(n)$ are the preprocessing time and operation time (maximum between query and update) of a dynamic, three-dimensional ray-shooting data structure.
\end{theorem}

\begin{proof}
Each iteration of the algorithm makes one ray-shooting query. At each iteration, either a motorcycle is added to the chain (Case~(d)), or at least one motorcycle is determined (Cases~(a---c)). 

Motorcycles begin as undetermined and, once they become determined, they remain so. This bounds the number of Cases~(a---c) to $n$. In Cases~(b) and~(c), one undetermined motorcycle may be removed from the chain. Thus, the number of undetermined motorcycles removed from the chain is at most $n$. It follows that Case~(d) happens at most $2n$ times.

Overall, the algorithm takes at most $3n$ iterations, so it needs no more than $3n$ ray-shooting queries and at most $n$ ``clipping'' updates where we change a triangular curtain into a slab. It follows that the runtime is $O(P(n)+nT(n))$.
\end{proof}
\ifthenelse{\boolean{long}}{In terms of space, we only need a linear amount besides the space required by the data structure.}{}

The previous best known algorithm runs in time $O(P(n)+n(T(n)+\log n)\log n)$~\cite{Vigneron2014}. Besides ray-shooting queries, it also uses range searching data structures, which do not increase the asymptotic runtime but make the algorithm more complex.

Agarwal and Matou\v{s}ek~\cite{agarwal93} give a ray-shooting data structure for curtains in $\R^3$ which achieves $P(n)=O(n^{4/3+\eps})$ and $T(n)=O(n^{1/3+\eps})$ for any $\eps>0$. Using this structure, both our algorithm and the algorithm of Vigneron and Yan~\cite{Vigneron2014} run in $O(n^{4/3+\eps})$ time for any $\eps>0$. If both algorithms use the same $\eps$ in the ray-shooting data structure, then our algorithm is asymptotically faster by a logarithmic factor.

\ifthenelse{\boolean{long}}{
\subsection{Special cases and remarks}\label{sec:mgremarks}
\subfile{sectionmotorcycleremarks}

}{
See Appendix~\ref{sec:mgremarks} for additional results in special cases and some remarks.
}

\ifthenelse{\boolean{long}}{
\section{Stable Matching Problems}\label{sec:kattr}

\subfile{sectionkattr}

\section{Server Cover}\label{sec:clientcover}
\subfile{sectionservercover}

}{}

\section{Conclusions}\label{sec:conclusions}
Before this paper, NNC had been used only in agglomerative hierarchical clustering and stable matching problems based on proximity. This paper adds the following use cases:
\begin{itemize}
    \item Its first use in problems without symmetric distances (motorcycle graphs, Section~\ref{sec:motorcycle}). We showed that the chain still works for finding nearest-neighbor cycles.
    \item Its first use without a NN structure (multi-fragment TSP, which uses our new soft NN structure, Section~\ref{sec:mfh}).
    \item Its first use in problems without global-local equivalence (server cover, \autoref{sec:clientcover}). We showed that the chain can be adapted in settings where matching MNNs is not as good as matching overall closest pairs.
    \item Its first use in approximation algorithms (also server cover).
    \item Its first use in a graph-theoretical framework (Steiner TSP, \autoref{sec:mfgraph}).
    \item Its first use in stable matching problems not based on distances (narcissistic $k$-attribute stable matching, \autoref{sec:kattr}).
\end{itemize}
The above applications illustrate the two main points of this paper: first, that in several geometric problems, finding mutually nearest neighbors leads to the same solution as finding closest pairs, which we call global-local equivalence; second, that MNNs can be found efficiently thanks to the NNC algorithm.

We expect that this algorithm will find more uses in computational geometry.
When dealing with problems involving nearest neighbors or closest pairs in some way, one may check if a form of global-local equivalence holds. If so, one should then consider using the NNC algorithm. The main guidelines for designing NNC algorithms are: (1) each link in the chain should get closer to MNNs; (2) to avoid infinite loops, the chain should be acyclic. One should be careful to break ties consistently; (3) after finding and processing MNNs, all the previous links in the chain should remain valid. These simple ingredients are likely to lead to an algorithm with a runtime of the form $O(P(n)+nT(n))$, as seen throughout this paper.

\ifthenelse{\boolean{long}}{
\noindent
We conclude with some open questions.
\begin{itemize}
    \item Can we use specialized data structures? Throughout the paper, we have used fully dynamic data structures that allow insertions and deletions. However, NNC algorithms typically only use deletions. Further, query points are generally known beforehand. Therefore, specialized data structures with these considerations in mind may speed up the algorithms in this paper, and NNC algorithms in general.
    \item Does the SNN structure have more uses? The SNN structure has a curious type of queries, where soft answers are not directly related to the query point. Nonetheless, we showed how it can be used as part of the SNNC algorithm and to solve the closest pair problem.
    \item Are other narcissistic stable matching models symmetric? As mentioned, narcissistic is a descriptive term for stable matching models where the preferences of each agent are determined by the agent's own qualities or attributes. We find it likely that NNC can be used for other narcissistic models, as such preferences seem unlikely to create cycles of first choices.
    \end{itemize}
}{}

\bibliography{biblio}

\ifthenelse{\boolean{long}}{}{

\appendix

\section{Stable Matching Problems}\label{sec:kattr}

\subfile{sectionkattr}
\section{Server Cover}\label{sec:clientcover}

\subfile{sectionservercover}
\section{Choice of Parameters}\label{sec:params}

\subfile{sectionparameters}
\section{Proof of Global-Local Equivalence in Multi-Fragment TSP}\label{app:mfgle}

We show the proof for Lemma~\ref{lem:mfgle}.
\proofmfgle

\section{A soft Nearest-Neighbor Structure for Paths}\label{app:snnpath}

\subfile{sectionsnnforpaths}
\section{Multi-fragment for Steiner TSP}\label{sec:mfgraph}

\subfile{sectionsteinertsp}
\section{Motorcycle Graphs: Special Cases and Remarks}\label{sec:mgremarks}

\subfile{sectionmotorcycleremarks}
\section{Algorithm Visualizations}\label{app:snapshots}

\figsnncmftsp
\figmc

}

\end{document}

%% file: sectionparameters.tex
We left open the question of finding valid parameters $(\eps^*,k^*)$.
This question is related to the \textit{kissing number} of the metric space, which is the maximum number of points that can be on the surface of a unit sphere all at pairwise distance $\geq 1$. For instance, it is well known that the kissing number is $6$ in $(\R^2,L_2)$ and $12$ in $(\R^3,L_2)$. It follows that, in $(\R^2,L_2)$, $(\eps^*=0,k^*=7)$ are valid parameters. Of course, we are interested in $\eps^*>0$. Thus, our question is more general in the sense that our points are not constrained to lie on a sphere, but in a shell (and, to complicate things, the width of the shell depends on the number of points).

\begin{lemma}\label{lem:kissing}
There are valid parameters in any metric space $(\R^\delta,L_p)$.
\end{lemma}

\begin{proof}
Consider a shell with inner radius $1$ and outer radius $1+c$, for some constant $c>0$. A set of points in the shell at pairwise distance ${\geq 1}$ corresponds to a set of disjoint balls of radius $1/2$ centered inside the shell. Consider the volume of the intersection of the shell with such a ball. This volume is lower bounded by some constant, $v$, corresponding to the case where the ball is centered along the exterior boundary.
Since the volume of the shell, $v_s$, is itself constant, the maximum number of disjoint balls of radius $1/2$ that fit in the shell is constant smaller than $v_s/v$. This is because no matter where the balls are placed, at least $v$ volume of the shell is inside any one of them, so, if there are more than $v_s/v$ balls, there must be some region in the shell inside at least two of them. This corresponds to two points at distance ${<1}$.

Set $k$ to be $v_s/v$, and $\eps$ to be the constant such that $(1+\eps)^k=1+c$. Then, $(\eps,k)$ are valid parameters for $(\R^\delta,L_p)$.
\end{proof}

The dependency of $k$-ANN structures on $1/\eps$ is typically severe. Thus, for practical purposes, one would like to find a valid pair of parameters with $\eps$ as big as possible. The dependency on $k$ is usually negligible in comparison, and, in any case, $k$ cannot be too large because the shell's width grows exponentially in $k$. Thus, we narrow the question to optimizing $\eps$: what is the largest $\eps$ that is part of a pair of valid parameters?

We first address the case of $(\R^2,L_2)$, where we derive the optimal value for $\eps$ analytically. We then give a heuristic, numerical algorithm for general $(\R^\delta,L_p)$ spaces.

\paragraph*{Parameters in $(\R^2,L_2)$.}
Let $\eps_\varphi\approx 0.0492$ be the number such that $(1+\eps_\varphi)^{10}=\varphi$, where $\varphi=\frac{1+\sqrt{5}}{2}$ is the golden ratio.
The valid parameters with largest $\eps$ for $(\R^2,L_2)$ are $(\eps^*<\eps_\varphi,k^*=10)$ ($\eps^*$ can be arbitrarily close to $\eps_\varphi$, but must be smaller).
This follows from the following observations.
\begin{itemize}
    \item The kissing number is $6$, so there are no valid parameters with $k< 6$.
    \item The thinnest annulus (i.e., 2D shell) with inner radius $1$ such that $10$ points can be placed inside at pairwise distance ${\geq 1}$ has outer radius $\varphi=(1+\eps_\varphi)^{10}$. Figure~\ref{fig:flowers}, top, illustrates this fact. In other words, if the outer radius is any smaller than $\varphi$, two of the $10$ points would be at distance ${<1}$. Thus, any valid pair with $k=10$ requires $\eps$ to be smaller than $\eps_\varphi$, but any value smaller than $\eps_\varphi$ forms a valid pair with $k=10$.
    \item For $6\leq k<10$ and for $k>10$, it is possible to place $k$ points at pairwise distance ${>1}$ in an annulus of inner radius $1$ and outer radius $(1+\eps_\varphi)^k$, and they are not packed ``tightly'', in the sense that $k$ points at pairwise distance ${>1}$ can lie in a thinner annulus. This can be observed easily; Figure~\ref{fig:flowers} (bottom) shows the cases for $k=9$ and $k=11$. Cases with $k<9$ can be checked one by one; in cases with $k>11$, the annulus grows at an increasingly faster rate, so placing $k$ points at pairwise distance ${>1}$ of each other becomes increasingly ``easier''. Thus, for any $k\not=10$, any valid pair with that specific $k$ would require an $\eps$ smaller than $\eps_\varphi$. 
\end{itemize}

\begin{figure}[t]
    \centering
    \includegraphics[width=0.95\linewidth]{flowers}
    \caption{\textbf{Top:} The first figure shows two concentric circles of radius $1$ and $\varphi$ with an inscribed pentagon and decagon, respectively, and some proportions of these shapes. The other figures show two different ways to place 10 points at pairwise distance ${\geq 1}$ inside an annulus of inner radius $1$ and outer radius $(1+\eps_\varphi)^{10}=\varphi$. Disks of radius $1/2$ around each point are shown to be non-overlapping. In one case, the points are placed on the vertices of the decagon. In the other, they alternate between vertices of the decagon and the pentagon. In both cases, the distance between adjacent disks is $0$. Thus, these packings are ``tight'', i.e., if the annulus were any thinner, there would be two of the $10$ points at distance $<1$. \textbf{Bottom:} $9$ and $11$ points at pairwise distance ${\geq 1}$ inside annuli of radius $(1+\eps_\varphi)^{9}$ and $(1+\eps_\varphi)^{11}$, respectively. These packings are not tight, meaning that, for $k=9$ and $k=11$, a valid value of $\eps$ would have to be smaller than $\eps_\varphi$.}
    \label{fig:flowers}
\end{figure}

\paragraph*{Parameters in \texorpdfstring{$(\R^\delta,L_p)$}{other metric spaces}.}
For other $L_p$ spaces, we suggest a numerical approach. We can do a binary search on the values of $\eps$ to find one close to optimal. For a given value of $\eps$, we want to know if there is any $k$ such that $(\eps,k)$ are valid. We can search for such a $k$ iteratively, trying $k=1,2,\ldots$ (the answer will certainly be ``no'' for any $k$ smaller than the kissing number). Note that, for a fixed $k$, the shell has constant volume. 
As in Lemma~\ref{lem:kissing}, let $v$ be the volume of the intersection between the shell and a ball of radius $1/2$ centered on the exterior boundary of the shell.  
As argued before, if $kv$ is bigger than the shell's volume, then $(\eps,k)$ are valid parameters. 
For the termination condition, note that if in the iterative search for $k$, $k$ reaches a value where the volume of the shell grows more than $v$ in a single iteration, no valid $k$ for that $\eps$ will be found, as the shell grows faster than the new points cover it.

Besides the volume check, one should also consider a lower bound on how much of the shell's surface (both inner and outer) is contained inside an arbitrary ball. We can then see if, for a given $k$, the amount of surface contained inside the $k$ balls is bigger than the total surface of the shell, at which point two balls surely intersect. This check finds better valid parameters than the volume one for relatively thin shells, where the balls ``poke'' out of the shell on both sides.

%% file: sectionsnnforpaths.tex
We simulate a SNN structure for paths with a SNN structure for points.
Given a set of paths, we maintain the set of path endpoints in the SNN structure for points.
Updates are straightforward: we add or remove both endpoints of the path.
Given a query path $q$ with endpoints $\{q_1,q_2\}$, we do a SNN query from each endpoint of the path. If both answers are hard (assuming that the path has two distinct endpoints, otherwise, just the one), then we find the true NN of the path, and we can return it. However, there is a complication with soft answers: the two points returned could be the endpoints of the same path. Thus, it could be the case that we find two closer points, but not two closer paths, as we need. The solution is to modify the specification of the SNN structure for points so that soft answers, instead of returning two points closer to each other than the query point to its NN, return three pairwise closer points. We call this a \textit{three-way} SNN structure. In the context of using the structure for paths, this guarantees that even if two of the three endpoints belong to the same path, at least two different paths are involved.

Lemma~\ref{lem:3waysnn} shows how to obtain a three-way SNN structure for points, Algorithm~\ref{alg:snnpath} shows the full algorithm for answering SNN queries about paths using a three-way SNN structure for points, and Lemma~\ref{lem:snnpath} shows its correctness.

\begin{lemma}\label{lem:3waysnn}
In any fixed dimension and for any $L_p$ metric, there is a three-way SNN structure with $O(n\log n)$ preprocessing time and $O(\log n)$ operation time (queries and updates).
\end{lemma}

\begin{proof}
Recall the implementation of the SNN structure from Section~\ref{sec:snn}. To obtain a three-way SNN structure, we need to change the values of $\eps$ and $k$ to make the shell smaller and $k$ bigger, so that if there are $k$ points in a shell of inner radius $1$ and outer radius $(1+\eps)^k$, then there must be at least three points at pairwise distance less than $1$. The method described in Section~\ref{sec:params} for finding valid parameters in $(R^\delta,L_p)$ also works here. It only needs to be modified so that the area (or surface) of the shell is accounted for twice. Since $k$ and $\eps$ are still constant, this does not affect the asymptotic runtimes in Theorem~\ref{thm:snn}.
\end{proof}

\begin{algorithm}
\caption{Soft-nearest-neighbor query for paths.}
\label{alg:snnpath}
\begin{algorithmic}
\State Let $q_1$ and $q_2$ be the endpoints of the query path.
\State Let $S$ be a three-way SNN structure containing the set of path endpoints.
\State Query $S$ with $q_1$ and $q_2$.
\If{both answers are hard}
\State Let $p_1$ and $p_2$ be the respective answers.
\State\Return the closest path to the query path among the paths with endpoints $p_1$ and $p_2$.
\ElsIf{one answer is hard and the other is soft}
\parState{Let $p$ be the hard answer to $q_1$ and $(a,b,c)$ the soft answer to $q_2$ (wlog). 
Let $P$ and $P'$ be the two closest paths among the paths with endpoints $a,b,$ and $c$.} \If{$d(p,q)<d(P,P')$}
\State\Return the path with endpoint $p$.
\Else
\State\Return $(P,P')$.
\EndIf
\Else $ $ (both answers are soft)
\State Let $(a_1,b_1,c_1)$ and $(a_2,b_2,c_2)$ be the answers to $q_1$ and $q_2$.
\State\Return the closest pair of paths among the paths with endpoints $a_1,b_1,c_1,a_2,b_2,c_2$. 
\EndIf 
\end{algorithmic}
\end{algorithm}

\begin{lemma}\label{lem:snnpath}
In any fixed dimension, and for any $L_p$ metric, we can maintain a set of $n$ paths in a SNN structure for paths with $O(n\log n)$ preprocessing time and $O(\log n)$ operation time (queries and updates).
\end{lemma}

\begin{proof}
All the runtimes follow from Lemma~\ref{lem:3waysnn}:
we maintain the set of path endpoints in a three-way SNN structure $S$. The structure $S$ can be initialized in $O(n\log n)$ time. Updates require two insertions or deletions to $S$, taking $O(\log n)$ time each. Algorithm~\ref{alg:snnpath} for queries clearly runs in $O(\log n)$ time. We argue that it returns a valid answer. Let $q$ be a query path with endpoints $\{q_1,q_2\}$, and consider the three possible cases:
\begin{itemize}
    \item Both answers are hard. In this case, we find the closest path to each endpoint, and, by definition, the closest of the two is the NN of $q$.
    \item One answer is soft and the other is hard. Let $p$ be the hard answer to $q_1$ and $(a,b,c)$ the soft answer to $q_2$ (wlog). Let $P$ and $P'$ be the two closest paths among the paths with endpoints $a,b,$ and $c$. If $d(p,q)<d(P,P')$, then, the path with $p$ as endpoint must be the NN of $q$, because there is no endpoint closer than $d(P,P')$ to $q_2$. Otherwise, $P,P'$ is a valid soft answer, as they are closer to each other than either endpoint of $q$ to their closest endpoints. 
    \item Both answers are soft. Assume (wlog) that the NN of $q$ is closer to $q_1$ than $q_2$. Then, the soft answer to $q_1$ gives us two paths closer to each other than $q$ to its NN, so we return a valid soft answer.
    \end{itemize}
\end{proof}

%% file: sectionsteinertsp.tex
In the traditional, non-Euclidean setting, a TSP instance consists of a complete graph with arbitrary distances. We remark that global-local equivalence (Lemma~\ref{lem:mfgle}) still holds in this general setting. In this context, the nearest neighbor of a path can be found in $O(n)$ time by iterating through the adjacency lists of both endpoints, where $n$ is the number of nodes. Using this linear search, we can easily compute the multi-fragment tour in $O(n^2)$ time with a NNC-based algorithm. It is a simpler version of Algorithm~\ref{alg:snccpath} that only has to handle hard answers and does not need any sophisticated data structures. This improves upon the natural way to implement the multi-fragment heuristic, which is to sort the $\Theta(n^2)$ edges by weight. Sorting requires $\Theta(n^2\log n)$ time. 

This is the first use of NNC in a graph-theoretical setting, but the fact of the matter is that the NNC algorithm can be used in any setting where we can find nearest neighbors efficiently.
Consider the related Steiner TSP problem~\cite{Cornuejols1985}: given a weighted, undirected graph and a set of $k$ node sites $P\subseteq V$, find a minimum-weight tour (repeated vertices and edges allowed) in $G$ that goes at least once through every site in $P$. Nodes not in $P$ do not need to be visited. For instance, $G$ could represent a road network, and the sites could represent the daily drop-off locations of a delivery truck. See~\cite{DEKOSTER2007481,ZHANG201530} for more applications.

Recently, Eppstein et al.~\cite{Eppstein17Latin} gave a NN structure for graphs from graph families with sublinear separators, which is the same as the class of graphs with polynomial expansion~\cite{dvorak2016}. For instance, planar graphs have $O(\sqrt{n})$-size separators\footnote{Other important families of sparse graphs with sublinear separators include $k$-planar graphs~\cite{DujEppWoo-SJDM-17}, bounded-genus graphs~\cite{gilbert1984}, minor-closed graph families~\cite{kawarabayashi2010}, and graphs that model road networks (better than, e.g., $k$-planar graphs)~\cite{eppstein2017crossing}.}. This data structure maintains a subset of nodes $P$ of a graph $G$, and, given a query node $q$ in $G$, returns the node in $P$ closest to $q$. It allows insertions and deletions to and from the set $P$. We cite their result in Lemma~\ref{lem:latin}.

\begin{lemma}[\cite{Eppstein17Latin}]\label{lem:latin}
Given an $n$-node weighted graph from a graph family with separators of size $S(n)=n^c$, with $0<c<1$, which can be constructed in $O(n)$ time, there is a dynamic\footnote{They (\cite{Eppstein17Latin}) use the term \textit{reactive} for the data structure instead of dynamic, to distinguish from other types of updates, e.g., edge insertions and deletions.} nearest-neighbor data structure requiring $O(nS(n))$ space and preprocessing time and that answers queries in $O(S(n))$ time and updates in $O(S(n)\log n)$ time.
\end{lemma}

As mentioned, one way to implement the multi-fragment heuristic is to sort the $\binom{k}{2}$ pairs of sites by increasing distances, and process them in order: for each pair, if the two sites are endpoints of separate paths, connect them.
The bottleneck is computing the distances. Running Dijkstra's algorithm from each site in a sparse graph, this takes $O(k(n\log n))$ (or $O(kn)$ in planar graphs~\cite{HENZINGER19973}). When $k$ is $\Theta(n)$, this becomes $O(n^2\log n)$. We do not know of any prior faster algorithm to compute the multi-fragment tour for Steiner TSP.

Since we have global-local equivalence, we can use the NNC algorithm to construct the multi-fragment tour in $O(P(n)+kT(n))$ time, where $P(n)$ and $T(n)$ are the preprocessing and operation time of a nearest-neighbor structure. 
Thus, using the structure from~\cite{Eppstein17Latin}, we get:

\begin{theorem}
The multi-fragment tour for the steiner TSP problem can be computed in $O(nS(n)+kS(n)\log n)$-time in weighted graphs from a graph family with separators of size $S(n)=n^c$, with $0<c<1$.
\end{theorem}


Finally, in graphs of bounded treewidth, which have separators of size $O(1)$, the data structure from~\cite{Eppstein17Latin} achieves $P(n)=O(n\log n)$ and $T(n)=O(\log^2 n)$, so we can construct a multi-fragment tour in $O(n\log n+k\log^2 n)$.

%% file: sectionmotorcycleremarks.tex
Consider the case where all motorcycles start from the boundary of a simple polygon with $O(n)$ vertices, move through the inside of the polygon, and also crash against the edges of the polygon. In this setting, the motorcycle trajectories form a connected planar subdivision. There are dynamic ray-shooting queries for connected planar subdivisions that achieve $T(n)=O(\log^2 n)$~\cite{Goodrich1993}. Vigneron and Yan used this data structure in their algorithm to get a $O(n\log^3 n)$-time algorithm for this case~\cite{Vigneron2014}. Our algorithm brings this down to $O(n\log^2 n)$. Furthermore, their other data structures require that coordinates have $O(\log n)$ bits, while we do not have this requirement.

Vigneron and Yan also consider the case where motorcycles can only go in $C$ different directions. They show how to reduce $T(n)$ to $\min(O(C\log^2 n, C^2\log n)$, leading to a $O(n\log^2 n C\min(\log n, C))$ algorithm for motorcycle graphs in this setting. Using the same data structures, the NNC algorithm improves the runtime to $O(n\log n C\min(\log n, C))$.

A remark on the use of our algorithm for computing straight skeletons: degenerate polygons where two shrinking reflex vertices collide gives rise to a motorcycle graph problem where two motorcycles collide head on. To compute the straight skeleton, a new motorcycle should emerge from the collision. Our algorithm does not work if new motorcycles are added dynamically (such a motorcycle could, e.g., disrupt a NN cycle already determined), so it cannot be used in the computation of straight skeletons of degenerate polygons.

As a side note, the NNC algorithm for motorcycle graphs is reminiscent of Gale's top trading cycle algorithm~\cite{shapley1974cores} from the field of economics. That algorithm also works by finding ``first-choice'' cycles. We are not aware of whether they use a NNC-type algorithm to find such cycles; if they do not, they certainly can; if they do, then at least our use is new in the context of motorcycle graphs.

%% file: sectionkattr.tex
We introduce the \textit{narcissistic k-attribute stable matching} problem, a special case of \textit{k-attribute stable matching}, and show that it belongs to the class of \textit{symmetric stable matching} problems. We use this fact to give an efficient NNC-type algorithm for it.

\subparagraph*{Stable matching.}
The stable matching problem studies how to match two sets of agents in a market where each agent has its own preferences about the agents of the other set 
in a ``stable'' manner.
Some of its applications include matching hospitals and residents~\cite{thematch} and on-line advertisement auctions~\cite{Aggarwal2009}.
It was originally formulated  by Gale and Shapley~\cite{gale62}
in the context of establishing marriages between $n$ men and $n$ women,
where each man ranks the women and the women rank the men.
A matching between the men and women is \emph{stable}
if there is no \emph{blocking pair}: 
a man and woman who prefer each other over their assigned choices. 

Gale and Shapley~\cite{gale62} showed that a stable solution 
exists for any set of preferences (and it might not be unique), and presented the deferred-acceptance algorithm, which finds a stable matching in
$O(n^2)$ time.

\subsection{Restricted models}
For arbitrary preference lists, Gale--Shapley's deferred-acceptance algorithm is worst-case optimal, as storing all the preferences already requires $\Theta(n^2)$ space (quadratic lower bounds are known also for ``simpler'' questions, like verifying stability of a \emph{given} matching~\cite{Gonczarowski2015}). This inspired work on finding subquadratic algorithms in restricted settings where preferences can be specified in subquadratic space. Such models are collectively called succinct stable matching~\cite{moeller2016}. We introduce a new model which is a special case of the following three models (of which none is a special case of another):
\begin{description}
\item[$k$-attribute model~\cite{bhatnagar2008}.] Each agent $p$ has a vector $\vec{p_a}$ of $k$ numerical attributes, and a vector $\vec{p_w}$ of $k$ weights according to how much $p$ values each attribute in a match. Then, each agent $p$ ranks the agents in the other set according to the objective function $f_p(q)=\vec{p_w}\cdot\vec{q_a}$, the linear combination of the attributes of $q$ according to the weights of $p$.
\item[Narcissistic stable matching.] This term is used to describe models where the preferences of each agent reflect their own qualities in some way (e.g., in~\cite{chennumber,moeller2016}).
\item[Symmetric stable matching~\cite{eppstein2017_2}.] Consider the setting where each agent $p$ has an arbitrary objective function, $f_p(q)$, and ranks the agents according to $f_p(q)$ (note that any set of preference lists can be modeled in this way). The preferences are called \textit{symmetric} if for any two agents $p,q$ in different sets, $f_p(q)=f_q(p)$.
\end{description}

In this paper, we consider the natural narcissistic interpretation of the $k$-attribute model, where $\vec{p_a}=\vec{p_w}$. That is, each agent weighs each attribute according to its own value in that attribute. To illustrate this model, consider a centralized dating service where $k$ attributes are known for each person, such as income, intelligence, and so on. In the general $k$-attribute model, each person assigns weights to the attributes according to their preferences. The narcissistic assumption that $\vec{p_a}=\vec{p_w}$ implies that someone with, say, a high income, values income more than someone with a relatively smaller income.

We make a general position assumption that there are no ties in the preference list of each agent. 
In addition, in this model each agent is uniquely determined by its attribute vector, so we do not distinguish between the agents themselves and their $k$-dimensional vectors. We obtain the following formal problem.

\begin{definition}[Narcissistic $k$-attribute stable matching problem]
Find a stable matching between two sets of $n$ vectors in $\R^k$, where a vector $\vec{p}$ prefers $\vec{q}$ over $\vec{q'}$ if and only if $\vec{p}\cdot \vec{q}>\vec{p}\cdot \vec{q'}$.
\end{definition}

We give a $O(n^{2-4/(k+2+\eps)})$-time algorithm for the problem.
Without the narcissistic assumption, the $k$-attribute model becomes less tractable: Künnemann et al.~\cite{moeller2016} showed that no strongly subquadratic-time algorithm exists if $k=\omega(\log n)$ assuming the Strong Exponential Time Hypothesis, even if the weights and attributes take Boolean values. (Similarly to us, \cite{moeller2016} also studied some restricted cases and presented a $O(C^{2k}n(k+\log n))$-time algorithm for the case where attributes and weights may have only $C$ different values and
a $\tilde{O}(n^{2-1/\lfloor k/2\rfloor})$-time algorithm for the asymmetric case where one of the sets has a single attribute and the other has $k$.\footnote{The $\tilde{O}$ notation ignores logarithmic factors.}) 

It is easy to see that our setting is symmetric: since $\vec{p_a}=\vec{p_w}$, for any two agents $p,q$ we have $f_p(q)=\vec{p_w}\cdot\vec{q_a}=\vec{p_a}\cdot\vec{q_w}=f_q(p)$. Eppstein et al.~\cite{eppstein2017_2} showed that in symmetric models, the NNC algorithm can be used.
Specifically, they introduced symmetric stable matching as an abstraction of the case where the agents are points in a metric space and they rank the agents in the other set by proximity~\cite{arkin2009geometric}. They showed that if preferences are symmetric, the problem has special properties: there is a unique stable matching and it can be found by repeatedly matching the two unmatched elements with the highest objective function value.
In addition, they showed the global-local equivalence: it suffices to match any two elements who have each other as first choice, called \textit{soul mates}, which are
mutually nearest neighbors if the preferences are distance-based. 

The algorithm of Eppstein et al.~\cite{eppstein2017,eppstein2017_2} for symmetric stable matching was the first use of NNC outside of hierarchical clustering. The algorithm is a bichromatic version of the NNC algorithm, where each individual in the chain is followed by its first choice among the unmatched individuals in the other set. Following such a chain inevitably leads to soul mates, which are then matched and removed permanently (here, the symmetry assumption is the key to avoid cycles in the chain). 
The execution relies on a dynamic \textit{first-choice data structure}, which maintains the elements in one set, and, given a query element $q$ from the other set, returns the first choice of $q$ among the elements in the structure. The final result is as follows:

\begin{lemma}[\cite{eppstein2017_2}]\label{lem:nncforssm}
Given a first-choice data structure with $P(n)$ preprocessing time and $T(n)$ operation time (maximum between query and update), a symmetric stable matching problem can be solved in $O(P(n)+nT(n))$.
\end{lemma}

\subsection{Our algorithm for narcissistic \texorpdfstring{$k$}{k}-attribute stable matching}

Adapting the NNC algorithm from~\cite{eppstein2017_2} to our model simply requires using an appropriate structure for first-choice queries. In our case, the first-choice data structure should maintain a set of vectors, and, given a query vector, return the vector maximizing the dot product with the query vector.
In the dual, this becomes ray shooting: each vector becomes a hyperplane, and a query asks for the first hyperplane hit by a vertical ray from the query point.
We use the data structure from~\cite{Matousek1992}, the runtime of which is captured in the following lemma (see~\cite{pankaj99range} for a summary of ray-shooting data structures).

\begin{lemma} \emph{(\cite[Theorem~1.5]{Matousek1992}).} \label{lem:rayshooting}
Let $\eps>0$ be a constant, $k\geq 4$ a fixed dimension, and $m$ a parameter with $n\leq m\leq n^{\lfloor k/2\rfloor}$. Then, there is a dynamic data structure for ray-shooting queries with $O(m^{1+\eps})$ space and preprocessing time, $O(m^{1+\eps}/n)$ update time, and $O(\frac{n}{m^{1/\lfloor k/2\rfloor}}\log n)$ query time.
\end{lemma}

\begin{theorem}\label{thm:kattr}
For any $\eps>0$,
the narcissistic $k$-attribute stable matching problem can be solved in $O(n\log n)$ time for $k=2$, $O(n^{4/3+\eps})$ time for $k=3$, and $O(n^{2-4/(k(1+\eps)+2)})$ time for $k\geq 4$. 
\end{theorem}

\begin{proof}
Since the problem is symmetric, it can be solved in $O(P(n)+nT(n))$ time, given a dynamic data structure for ray-shooting queries with $P(n)$ preprocessing time and $T(n)$ operation time (Lemma~\ref{lem:nncforssm}). 

For $k\geq 4$, using the data structure for ray-shooting queries from~\cite{Matousek1992} (Lemma~\ref{lem:rayshooting}) results in a runtime of $O(m^{1+\eps}+\frac{n^2\log n}{m^{1/\lfloor k/2\rfloor}})$ for any $\eps>0$. The optimal runtime is achieved when the parameter $m$ is chosen to balance the two terms, i.e., so that $m^{1+\eps}=\frac{n^2\log n}{m^{1/\lfloor k/2\rfloor}}$. This gives $m=(n^2\log n)^{1/(1+\eps+1/\lfloor k/2\rfloor)}$. For the sake of obtaining a simple asymptotic expression, we set $m$ to $(n^2\log n)^{1/(1+\eps+2/k)}$ (which is the same for even $k$, and bigger for odd $k$). Then, the $O(m^{1+\eps})$ term dominates. Also note that if $\eps<1-2/k$, this value of $m$ is between $n$ and $n^{\lfloor k/2\rfloor}$, so the condition in Lemma~\ref{lem:rayshooting} is satisfied.

Thus, the problem can be solved in $O(m^{1+\eps})=O((n^2\log n)^{(1+\eps)/(1+\eps+2/k)})$, which further simplifies to the claimed runtime of $O(n^{2-4/(k(1+\eps')+2)})$ (where $\eps'$ needs to satisfy $\eps'>\eps$).
For $k=3$, we use the same data structure, but raising the problem to four dimensions, so that Lemma~\ref{lem:rayshooting} applies. For $k=2$, see Lemma~\ref{lem:2attr}.
\end{proof}

Incidentally, the value for $m$ used in Theorem~\ref{thm:kattr} also improves the algorithm by Künnemann et al.~\cite{moeller2016} for the one-sided $k$-attribute stable matching problem \cite[Theorem~2]{moeller2016}, which also depends on the use of this data structure. The improvement is from $\tilde{O}(n^{2-1/(k/2)})$ to $O(n^{2-1/(k(1+\eps)/4+1/2)})$.
Similar balancing of preprocessing and query times in~\cite[Corollary~5.2]{Mato1992trees} also improves the time to verify stability of a \emph{given} matching in the (2-sided) $k$-attribute stable matching model~\cite[Section~5.1]{moeller2016} for constant $k$; the improvement is from $\tilde{O}(n^{2-1/(2k)})$ to $O(n^{2-1/(k/2+1/2+\eps)})$ for any $\eps>0$.

\subsection{The 2-attribute case}
In this special case, we can design a simple first-choice data structure with $P(n)=O(n\log n)$ preprocessing time and $T(n)=O(\log n)$ operation time.
Note that, for a vector $\vec{p}$ in $\R^2$, all the points along a line perpendicular to $\vec{p}$ are equally preferred, i.e., have the same dot product with $\vec{p}$ (because their projections onto the supporting line of $\vec{p}$ are the same). In fact, the preference list for $\vec{p}$ corresponds to the order in which a line perpendicular to $\vec{p}$ encounters the vectors in the other set as it moves in the direction opposite from $\vec{p}$ (see Figure~\ref{fig:kattr}, left).
We get the following lemma (where the vectors in one set are interpreted as points).

\begin{figure}
\centering
\includegraphics[width=.8\linewidth]{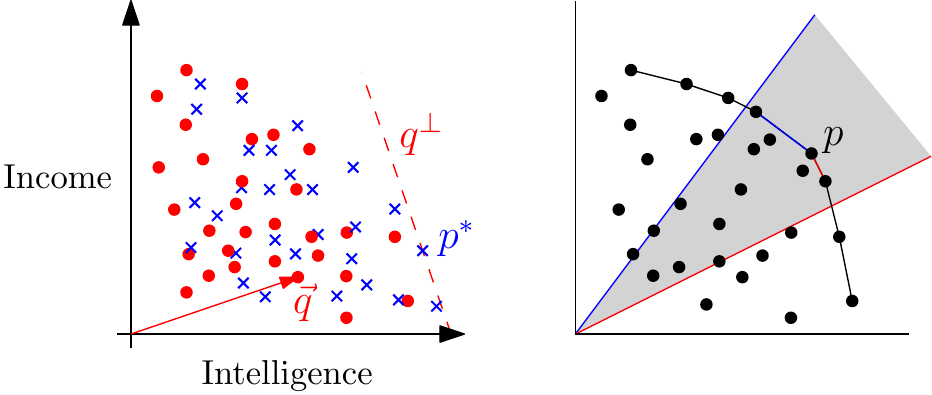}
\caption{\textbf{Left:} an instance of narcissistic $2$-attribute stable matching. The two sets of vectors are represented as red dots and blue crosses, respectively, in a plane where the axes correspond to the two attributes. For a specific red vector, $\vec{q}$, its first choice in the other set (the vector maximizing the dot product), $p^*$, is shown. The dashed line passing through $p^*$ is perpendicular to $\vec{q}$. \textbf{Right:} the point $p$ is the point among the black points maximizing $q\cdot p$ for all the points $q$ in the gray wedge. The wedge is delimited by two rays starting at the origin and perpendicular to the two edges of the convex hull incident to $p$.}
\label{fig:kattr}
\end{figure}

\begin{lemma}\label{lem:ch}
Given a point set $P$ and a vector $\vec{q}$, in $\R^2$, the point $p^*$ in $P$ maximizing $\vec{q}\cdot p^*$ is in the convex hull of $P$.  
\end{lemma}
\begin{proof}
Consider a line perpendicular to $\vec{q}$. Move this line
in the direction of $\vec{q}$, until all points in $P$ lie on the same side of it (behind it). Note that any line orthogonal to $\vec{q}$ has the property that all points lying on the line have the same dot product with $\vec{q}$.
The point $p^*$ is the last point in $P$ to touch the line, since moving the line in the opposite direction from $\vec{q}$ decreases the dot product of $\vec{q}$ with any point on the line (and by the general position assumption, it is unique). Clearly, $p^*$ is in the convex hull.
\end{proof}

Our first-choice data structure is a semi-dynamic convex hull data structure, where deletions are allowed but not insertions~\cite{Hershberger1992}.
We handle queries as in Lemma~\ref{lem:findfirstchoice}. 

\begin{lemma}\label{lem:findfirstchoice}
Given the ordered list of points along the convex hull of a point set $P$, and a query vector $\vec{q}$, we can find the point $p^*$ in $P$ maximizing $\vec{q}\cdot p^*$ in $O(\log n)$ time, where $n$ is the number of points in the convex hull.
\end{lemma}

\begin{proof}
By Lemma~\ref{lem:ch}, the point $p^*$ is in the convex hull.
For ease of exposition, assume that all the points in $P$ and $\vec{q}$ have positive coordinates (the alternative cases are similar).
Then, $p^*$ lies in the top-right section of the convex hull (the section from the highest point to the rightmost point, in clockwise order).

Note that points along the top-right convex hull are ordered by their $y$-coordinate, so, we say \emph{above} and \emph{below} to describe the relative positions of points in it.
Each point $p$ in the top-right convex hull is the point in $P$ maximizing $p\cdot \vec{q'}$ for all the vectors $\vec{q'}$ in an infinite wedge, as depicted in Figure~\ref{fig:kattr}, right. The wedge contains all the vectors $\vec{q'}$ whose perpendicular line touches $p$ last when moving in the direction of $\vec{q'}$, so the edges of the wedge are perpendicular to the edges of the convex hull incident to $p$. Thus, by looking at the neighbors of $p$ along the convex hull, we can calculate this wedge and know whether $\vec{q}$ is in the wedge for $p$, below it, or above it. Based on this, we discern whether the first choice of $\vec{q}$ is $p$ itself or above or below it. Thus, we can do binary search for $p^*$ in $O(\log n)$ time.
\end{proof}

\begin{lemma}\label{lem:2attr}
The narcissistic 2-attribute stable matching problem can be solved in $O(n\log n)$ time.
\end{lemma}

\begin{proof}
We can use the NNC algorithm from~\cite{eppstein2017_2} coupled with a first-choice data structure which is a semi-dynamic convex hull data structure. Updating the convex-hull can be done in $O(n\log n)$ time throughout the algorithm~\cite{Hershberger1992}. Queries are answered in $O(\log n)$ time (Lemma~\ref{lem:findfirstchoice}). Thus, the total running time is $O(n\log n)$. 
\end{proof}

%% file: sectionservercover.tex
\documentclass[main.tex]{subfiles}

Geometric \textit{coverage} problems deal with finding optimal configurations of a set of geometric shapes that contain or ``cover'' another set of objects (for instance, see~\cite{agarwal2014near,bronnimann1995almost,PedersenWang18}).
In this section, we propose an NNC-type algorithm for a problem in this category. We use NNC to speed up a greedy algorithm for a one-dimensional version of a \textit{server cover} problem: given the locations of $n$ clients and $m$ servers, which can be seen as houses and telecommunication towers, the goal is to assign a ``signal stregth'' to each communication tower so that they reach all the houses, minimizing the cost of transmitting the signals.

Formally, we are given two sets of points in $\R^\delta$, $S$ (servers) and $C$ (clients). The problem is to assign a radius, $r_i$, to a disk centered at each server $s_i$ in $S$, so that every client is contained in at least one disk. The optimization function to minimize is $\sum r_i^\alpha$ for some parameter $\alpha>0$. The values $\alpha=1$ and $\alpha=2$ are of special interest, as they correspond to minimizing the sum of radii and areas (in 2D), respectively.

\subsection{Related work}
Table~\ref{tab:cover} gives an overview of exact and approximation algorithms for the server cover problem. It shows that when either the dimension $\delta$ or $\alpha$ are larger than $1$, there is a steep increase in complexity.
We focus on the case with $\delta=1$ and $\alpha = 1$, which has received significant attention because it gives insight into the problem in higher dimensions.

Server coverage was first considered in the one-dimensional setting by Lev-Tov and Peleg~\cite{LevTov05}. They gave an $O((n + m)^3)$-time dynamic-programming algorithm for the $\alpha = 1$ case, where $n$ is the number of clients and $m$ is the number of servers. They also gave a linear-time 4-approximation (assuming a sorted input).
The runtime of the exact algorithm was improved to $O((n + m)^2)$ by Biniaz et al.~\cite{BiniazApprox}.
In the approximation setting, Alt et al.~\cite{carrots} gave a linear-time 3-approximation and an $O(m + n\log m)$-time $2$-approximation (also assuming a sorted input).
Using NNC, we improve this to a linear-time $2$-approximation algorithm under the same assumption that the input is sorted. 
\begin{table}[t]
    \begin{center}
            \begin{tabular}{c|c|c|c}
         \textbf{Dim.} & $\boldsymbol{\alpha}$       & \textbf{Approximation ratio} & \textbf{Complexity} \\\hline
         2D   & $\alpha > 1$   & Exact           & NP-hard~\cite{carrots}  \\\hline
         2D   & $\alpha = 1$   & Exact$^*$       & $O((n+m)^{881}T(n+m))$~\cite{GibsonClustering} \\
              &                & $1+\eps$ & $O((n+m)^{881}T(n+m))$~\cite{GibsonClustering} \\
              &                & $(1 + 6/k)$     & $O(k^2(nm)^{\gamma + 2})$~\cite{LevTov05} \\\hline
         1D   & $\alpha\geq 1$ & Exact           & Polynomial (high complexity)~\cite{Bilo2005} \\\hline
         1D   & $\alpha = 1$   & Exact           & $O((n + m)^2)$~\cite{BiniazApprox} \\
              &                & 3               & $O(n + m)$~\cite{carrots} \\
              &                & 2               & $O(m + n\log m)$~\cite{carrots} \\
              &                & 2               & $O(n + m)$ (this paper)
    \end{tabular}
        \end{center}
        \caption{Summary of best known results on the server cover problem. In the table: $n$ is the number of clients, $m$ is the number of servers, $\eps>0$ is an arbitrarily small constant, $k>1$ is an arbitrary integer parameter, $\gamma>0$ is shown to be a constant, and $T(n)$ is the cost of comparing the cost of two sets of disks in polynomial time, which requires comparing sums of square roots to compute exactly. $^*$The exact algorithm of~\cite{GibsonClustering} is under the assumption that $T(n)$ can be computed, which depends on the computational model.}
\label{tab:cover}
\end{table}

\subsection{Global-local equivalence in server cover}
The $O(m + n\log m)$-time $2$-approximation by Alt et al.~\cite{carrots} can be described as follows: start with disks (which, in 1D, are intervals) of radius $0$, and, at each step, make the smallest disk growth which covers a new client. If we define the distance $d(c,s)$ between a client $c$ and a server $s$ with a disk with radius $r\geq 0$ as the distance between $c$ and the closest boundary of the server's disk, the process can be described as repeatedly finding the closest uncovered client--server pair and growing the server's disk up to the client. Under this view, there is a natural notion of MNNs: an uncovered client $c$ and a server $s$ such that $d(c,s)$ is the smallest among all the distances involving $c$ and $s$.

However, Figure~\ref{fig:nogle} illustrates that this problem does not satisfy global-local equivalence: matching MNNs does not yield the same result as matching the closest pair. Furthermore, it shows that matching MNNs loses the 2-approximation guarantee.
We nevertheless use NNC to achieve a 2-approximation, which requires enhancing the algorithm so that it does not simply match MNNs. This shows that NNC may be useful even in problems where global-local equivalence does not hold.

\begin{figure}[t]
    \centering
    \includegraphics[width=0.99\textwidth]{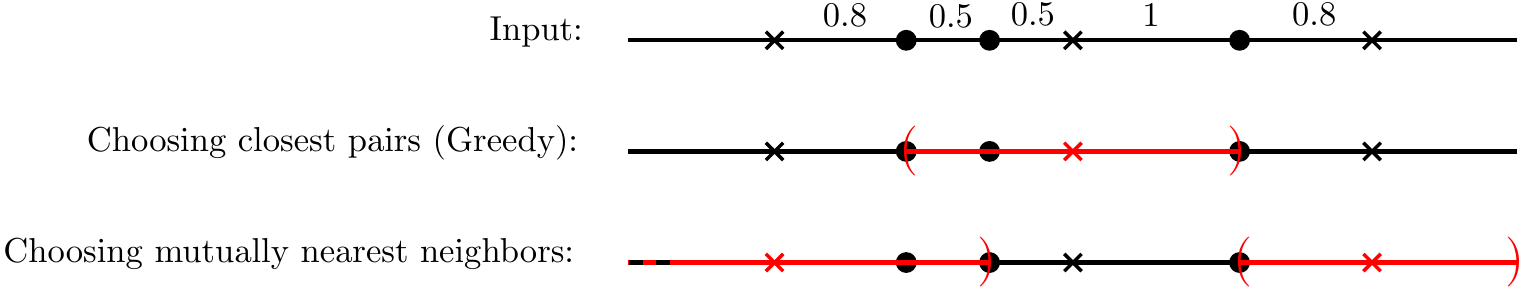}
    \caption{An instance where choosing MNNs in a specific order does not result in the same solution as choosing closest pairs (servers are crosses, clients are dots). Furthermore, note that the cost of the solution choosing MNNs, $2.1$, is not within a factor 2 of the optimal cost, $1$.}
    \label{fig:nogle}
\end{figure}

\subsection{Linear-time 2-approximation in 1D}

The algorithm takes a list of $n$ clients and $m$ servers ordered left-to-right, and outputs a radius for each server (which might be $0$).
In the algorithm, we group clients and servers into clusters. Each element starts as a base cluster, and we repeatedly merge them until there is a single cluster left. We distinguish between client clusters, consisting of a set of still uncovered clients, and server clusters, consisting of servers and covered clients. Clusters span intervals (as defined below). The distance $d(a,b)$ between clusters is defined as the distance between the closest endpoints of the clusters' intervals. We begin by describing the merging operation based on the cluster types. 
\begin{itemize}
    \item We merge client clusters into larger client clusters. All the clients in a cluster are eventually covered together, so we only need to keep track of the left-most one, $p$ and right-most one, $q$; thus, we represent the client cluster with the interval $[p,q]$. Each client $p$ starts as a cluster $[p,p]$. Two client clusters $[p,q]$ and $[p',q']$ (which in the algorithm never overlap), with $q<p'$ are merged into a client cluster $[p,q']$.
    \item We merge server clusters into larger server clusters. Of all the servers in a cluster, only the ones with disks reaching furthest to the left and to the right may cover new clients. Let these servers be $s_l$ and $s_r$, respectively (which might be the same), let $l$ be the left-most point covered by $s_l$, and $r$ the right-most point covered by $s_r$. Then, all the information we need about a server cluster is $([l,r],s_l,s_r)$. Note that $l \leq s_l \leq s_r \leq r$. Each server $s$ starts as a cluster $([s,s],s,s)$. To merge two server clusters $([p,q],s_p,s_q)$ and $([p',q'],s_{p'},s_{q'})$ (which may overlap), let $l^*=\min{(p,p')}$ and $r^*=\max{(q,q')}$. Replace both by a server cluster $([l^*,r^*],s_{l^*},s_{r^*})$. Retain the identities only of the two servers whose boundaries extend furthest left ($s_{l^*}$) and right ($s_{r^*}$).
    \item Merging a client cluster $[p,q]$ and a server cluster $([p',q'],s_{p'},s_{q'})$ (which may overlap) into a new server cluster involves covering all the clients in the cluster by $s\in\{s_{p'},s_{q'}\}$, whichever is cheaper.
    That is, Let $d^*$ be the new radius of the disk of $s$ after it grows to cover $[p,q]$; 
    we merge the client cluster and the server cluster into a server cluster $([l^*,r^*], s_{l^*}, s_{r^*})$, where $l^*=\min{(p',s-d^*)}$, $r^*=\max{(q',s+d^*)}$, and $s_{l^*}$ (resp. $s_{r^*}$) is the server among $s$ and $s_{p'}$ (resp. $s$ and $s_{q'}$) with the leftmost (resp. rightmost) extending disk.
\end{itemize}

The algorithm works by building a chain (a stack) of clusters ordered from left to right.
The following invariant holds at the beginning of each iteration: no two clusters overlap, the chain contains a prefix of the list of clusters, and the distance between successive clusters in the chain decreases. In the pseudocode (Algorithm~\ref{alg:clientcover}) we use $a\cup b$ to denote the cluster resulting from merging clusters $a$ and $b$, and $C(s)$ to denote the cluster containing a server $s$.

\begin{algorithm}[t]
\caption{Nearest-neighbor chain algorithm for 1D server cover with $\alpha=1$}
\label{alg:clientcover}
\begin{algorithmic}
\State Initialize the base client clusters and server clusters.
\State Initialize a stack (the chain) with the leftmost cluster.
\While{there is more than one cluster}
    \State Let $a$ be the cluster at the top of the chain, and $b$ its nearest neighbor.
    \If{$b$ is to the right of $a$}
        \State Add $b$ to the chain.
    \Else 
        \parState{Merge $a$ and $b$, remove them from the chain, and add $a\cup b$.}
        \If{a server $s$ grows to cover a client cluster $c$ as a result of the merge}
            \parState{(Note that the disk of $s$ grows on both sides of $s$. Thus, $C(s)$, which is $a\cup b$ at this point, might contain or overlap other clusters on the opposite side of $c$, as illustrated in Figure~\ref{fig:clientmerge}.)}
            \While{$C(s)$ is not disjoint from other clusters}
                \State Traverse the list of clusters from $C(s)$ in the opposite direction from $c$.
                \While{the next cluster, $e$, is contained in, or overlaps $C(s)$}
                    \parState{Merge $e$ and $C(s)$, remove them from the chain ($e$ might not be in the chain, if it is to the right of $s$, in which case only $C(s)$ is removed) and add the merged cluster to the chain.}
                \EndWhile
                \If{the last cluster $e$ partially overlaps $C(s)$}
                    \parState{Set $c$ to $e$. (Their merge may cause the disk of $s$ to expand on the opposite side from $e$, so $C(s)$ again might overlap with clusters on the opposite side.)}
                \Else
                    \State Break out of the while loop; $C(s)$ is disjoint from other clusters.
                \EndIf
            \EndWhile
        \EndIf
    \EndIf
\EndWhile
\end{algorithmic}
\end{algorithm}

\begin{figure}
    \centering
    \includegraphics[width=0.9\textwidth]{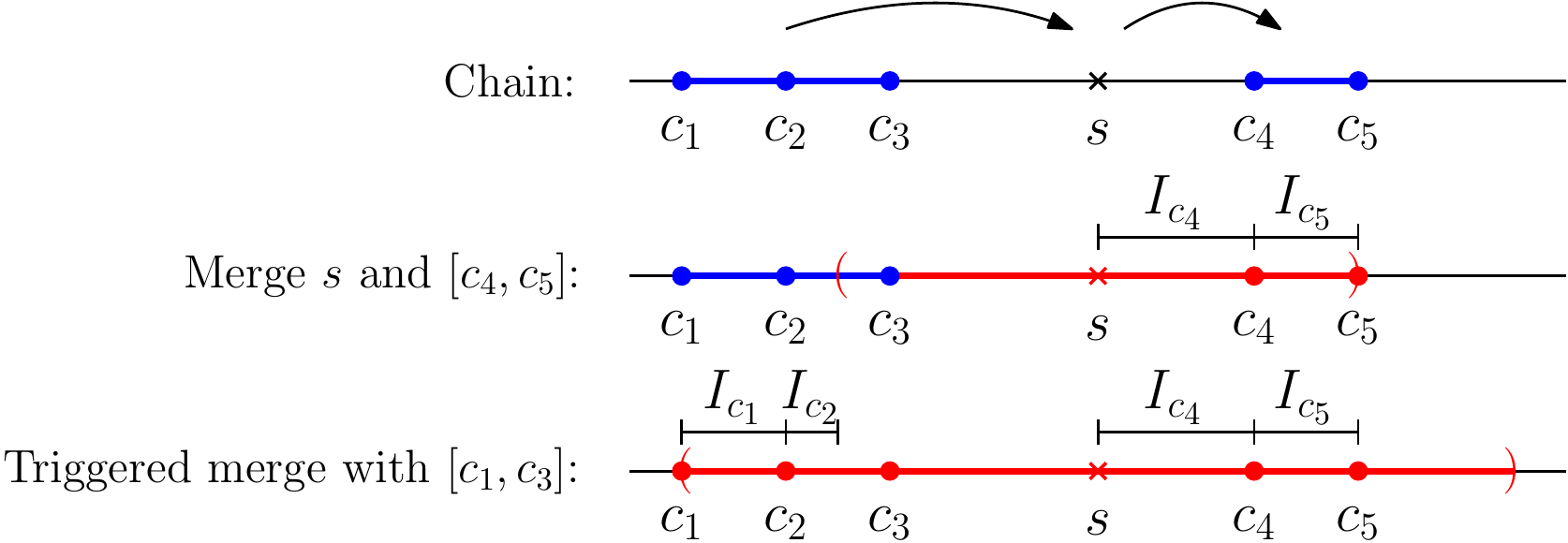}
    \caption{Illustration of the case where merging a server cluster and a client cluster causes the server cluster to expand on the opposite side and partially intersect another client cluster. This triggers another merge, causing the server cluster to expand again. The coverage intervals $I_c$ defined in the analysis are also shown.}
    \label{fig:clientmerge}
\end{figure}

\subsubsection{Correctness}
At the end of the algorithm, all the clusters have been merged into one, which is a server cluster (as long as there is one). Thus every client cluster has been merged with a server cluster, which means that some server grew its radius to cover the clients (or they became covered indirectly through an expansion). Thus, the output is a valid solution. We turn our attention to the analysis of the runtime and of the $2$-approximation factor. Throughout, we make an assumption that there are no ties between distances (or that they are broken consistently).

\begin{lemma}
Algorithm~\ref{alg:clientcover} runs in $O(n+m)$ time, assuming the input is given in sorted order.
\end{lemma}
\begin{proof}
Initially, there are $n+m$ clusters.
Each merge operation reduces the number of clusters by one, so the number of merge operations is $n+m-1$. A merge can be done in $O(1)$ time, so the total time spent doing merges is $O(m+n)$. Each iteration of the main loop either causes at least one merge or adds a new cluster to the chain. Since clusters stay in the chain until they are merged, this can only happen $O(n+m)$ times, so there are $O(n+m)$ iterations. In 1D, finding the NN of a cluster simply involves checking the previous and next clusters, which can be done in $O(1)$.
\end{proof}

Given an arbitrary problem instance, let NNC denote the solution output by Algorithm~\ref{alg:clientcover} and OPT denote an optimal solution.

\begin{theorem}\label{thm:clientcover}
$\mbox{cost}(NNC)\leq 2\mbox{cost}(OPT)$.
\end{theorem} 

We follow the proof idea for the greedy algorithm from~\cite{carrots}. We ``charge'' the disk radii in NNC to disjoint ``coverage intervals'', $I_c$, each of which is associated with a client $c$, and such that $\sum|I_c|=\mbox{cost}(\mbox{NNC})$, where $|I|$ denotes the length of an interval $I$. If the union of these intervals (and therefore the sum of their lengths) were entirely contained within the disks in OPT, then NNC would trivially be at most double the sum of radii in OPT. If that is not the case, we show that the length of every coverage interval outside of the disks in OPT is accounted for by an equal or greater absence of coverage intervals inside an OPT disk.

\begin{definition} Suppose that a server cluster $S$ and a client cluster $C$ are merged in Algorithm~\ref{alg:clientcover}, and, as a result, server $s$ is expanded to cover clients $c_1, \ldots, c_k$, in order of proximity to $s$. If $C$ is to the right of $s$, define the coverage interval $I_{c_1}$ as the open interval $(s_b,c_1)$, where $s_b$ is the right-most boundary of the disk of $s$ before the expansion, and define $I_{c_i}$ as $(c_i,c_{i-1})$ for $1 < i \leq k$. If $C$ is to the left, the intervals are defined symmetrically.
\end{definition}

See Figure~\ref{fig:clientmerge}, bottom, for an example of the coverage intervals.
To prove Theorem~\ref{thm:clientcover}, we need the following intermediate results.

\begin{lemma}\label{lem:disjoint}
The intervals $I_c,I_{c'}$ are disjoint if $c\not=c'$.
\end{lemma}

\begin{proof}
If $c$ and $c'$ belong to the same client cluster at the time $c$ is covered, it follows from the definition. Otherwise, let the cluster of $c$ be the one merged with a server cluster first. Then, after $c$ is covered, the interval $I_c$, if it exists, is inside a server cluster. Coverage intervals from clients covered later do not intersect existing server clusters.
\end{proof}

For a server $s$, let $D_O(s)$ denote the disk of $s$ in OPT.
To offset the intervals $I_c$ which occur outside of OPT disks, we need the following.

\begin{remark}\label{lem:imustoverlap}
Every $I_c$ intersects or has a shared endpoint with a disk $D_O(s)$.
\end{remark}
This is because $c$ must be covered by OPT.

Suppose that for some $c$, $I_c$ is not contained in any disk in OPT. Then, by Remark~\ref{lem:imustoverlap}, $I_c$ intersects or has a shared endpoint with a disk $D_O(s)$ in OPT. We consider the two possible cases separately, where $s$ is to the left or to the right of $c$. We show (Lemma~\ref{lem:leftrightcase}) that in either case there is an interval $J$, between $s$ and $c$ and inside $D_O(s)$, which is disjoint from all coverage intervals (Fig.~\ref{fig:clientproof}). Note that at most one $I_c$ may intersect $D_O(s)$ on each side on $s$, so the intervals $J$ do not overlap.

\begin{lemma}\label{lem:leftrightcase}
If a client $c$ belongs to $D_O(s)$ for some server $s$ and $I_c$ extends across the right (left) boundary of $D_O(s)$, 
then there is an interval $J$ in $D_O(s)$, to the right (left) of $s$, free of coverage intervals, and such that $|J|>|I_c|$. 
\end{lemma}

\begin{proof}
\textbf{Right case.}
Consider first the setting in Figure~\ref{fig:clientproof}, right: suppose that in OPT, $s$ covers some clients which in NNC are covered for the first time (the time where their coverage intervals are defined) from a server to the right of $D_O(s)$. Let $c_1,\ldots,c_k$, $k\geq 1$, be all such clients. Then, the coverage interval $I_{c_k}$ of $c_k$ extends across the right boundary of $D_O(s)$. 
Let $x$ be the input element (client or server, possibly $s$) immediately to the left of $c_1$, and $y$ the input element immediately to the right of $c_k$. Note that $d(c_k,y)\geq |I_{c_k}|$, since a coverage interval cannot extend past another input element.

We show that \textit{(i)} $d(x,c_1)>d(c_k,y)$ (and thus, $d(x,c_1)>|I_{c_k}|$), and that \textit{(ii)} the interval $(x,c_1)$ is free of coverage intervals.
Claim \textit{(i)} follows from the fact that if $d(x,c_1)<d(c_k,y)$, then $x$ and $c_1$ would be merged before $y$ is added to the chain, which cannot happen: if $x$ is a server, then $c_1$ would be covered from the left, and if $x$ is a client, $x$ and $c_1$ would be merged together, contradicting that $c_1$ is the left-most client covered from a server to the right of $D_O(s)$. For \textit{(ii)}, note that $c_1$ was covered from the right (by definition) and $x$ either was a client covered from the left (by definition of $c_1$) or a server which does not cover $c_1$. In the former case, $I_x$ has its right endpoint at $x$, and in the latter case, $x$ is not the client to cover $c_1$. Thus, there are no coverage intervals in $(x, c_1)$.

\textbf{Left case.}
Now consider the setting in Figure~\ref{fig:clientproof}, left. The setting is similar, except that $c_1,\ldots,c_k$, $k\geq 1$, are to the left of $s$ and are covered by a server to the left of $D_O(s)$, and it is the interval $I_{c_1}$ that extends across the left boundary of $D_O(s)$. 
Define $x$ as in the previous case but symmetrically: it is the input element immediately to the right of $c_k$. We define $y$ slightly differently: it is the right boundary of the cluster preceding $c_1$ at the time $c_1$ is added to the chain (not necessarily an input element). Note that $d(y,c_1)\geq |I_{c_1}|$, since a coverage interval for $c_1$ would start at, or to the right of $y$.

Let $u$ and $v$ be the two consecutive elements among $c_1,\ldots,c_k,x$ maximizing $d(u,v)$ (note that $x$ may be a server).
We show that \textit{(i)} $d(u,v)>d(y,c_1)$ (and thus, $d(u,v)>|I_{c_1}|$), and \textit{(ii)} $(u,v)$ is free of coverage intervals.
For \textit{(i)}, assume for a contradiction that $d(u,v)<d(y,c_1)$. Then, $c_1$ and $x$ (and all the elements in between) would be clustered together before they are merged with $y$, because $y$ would not be the NN of $c_1$ until all these merges between closer elements happen. However, this contradicts that $c_k$ is the right-most client covered for the first time from the left of $D_O(s)$. Therefore, we have \textit{(i)}.

For \textit{(ii)},
note that if $v$ is not $x$, then $v$ and $x$ (and all the elements in between) get clustered together before they are merged with $u$, because $u$ would not be the NN of $v$ until all these merges between closer elements happen. However, this contradicts that $c_k$ is the right-most client covered (for the first time) from the left of $D_O(s)$. Therefore, $v$ is $x$, and $c_k$ is $u$. \textit{(ii)} now follows by an analogous reasoning as in the symmetric case.
\end{proof}
We are ready to prove Theorem~\ref{thm:clientcover}.

\begin{figure}
    \centering
    \includegraphics[width=0.99\textwidth]{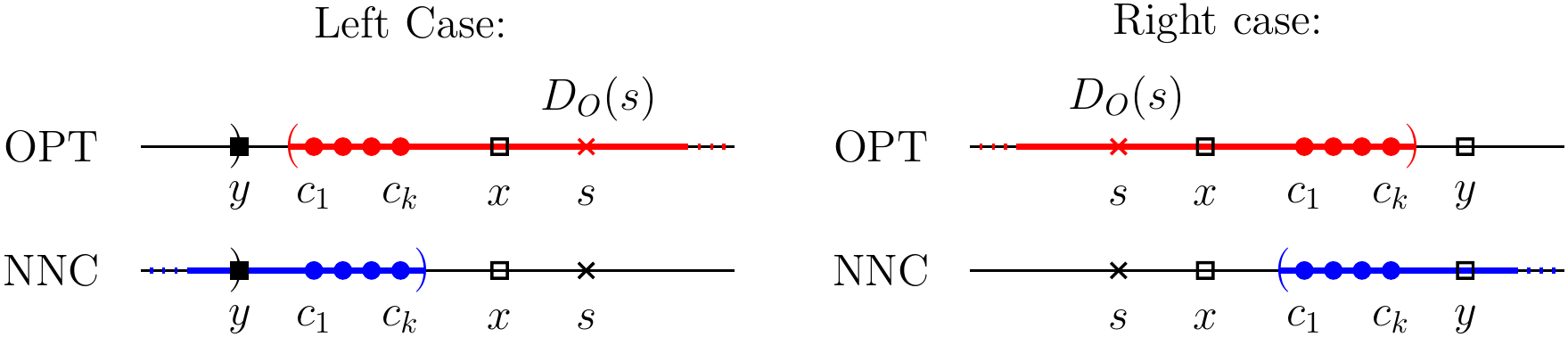}
    \caption{Illustration of the settings in the proof of Lemma~\ref{lem:leftrightcase}.} \label{fig:clientproof}
\end{figure}

\begin{proof}[Proof of Theorem~\ref{thm:clientcover}]

As mentioned, $\mbox{cost}(\mbox{NNC}) = \sum_c I_c$. The total length of the coverage intervals contained in OPT disks does not exceed twice the sum of the OPT radii (recall that by Lemma~\ref{lem:disjoint}, the coverage intervals are pairwise-disjoint). Consider now the parts of the coverage intervals outside the disks of OPT. Note that for each OPT disk, there is at most one interval $I_c$ overlapping the disk on each side, and by Remark~\ref{lem:imustoverlap}, every $I_c$ touches or overlaps a disk.
Furthermore, by Lemma~\ref{lem:leftrightcase}, for every length $\ell$ of coverage intervals in NNC outside the OPT disk of a server $s$ to the right (left) of $s$, there is at least $\ell$ length within $s$'s disk to the right (left) of $s$ that is free of coverage intervals. Therefore the approximation ratio of 2 is preserved.
\end{proof}

See Figure~\ref{fig:clienttight} for an instance that shows that the 2-approximation is tight. 

\begin{figure}
    \centering
    \includegraphics[width=0.6\textwidth]{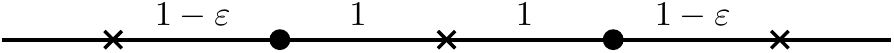}
    \caption{A tight example for the 2-approximation greedy algorithm of Alt et al.~\cite{carrots}. It is also a tight example for NNC, as $\mbox{cost}(NNC)=2-2\eps$, and $\mbox{cost}(OPT)=1$.}
    \label{fig:clienttight}
\end{figure}

\subsection{Greedy in higher dimensions}
As mentioned, the greedy algorithm which makes the smallest disk growth, at each step, which covers a new client, achieves a $2$-approximation in the 1D setting~\cite{carrots}. Does it achieve a good approximation ratio in higher dimensions? In this section, we give a negative answer. It performs poorly in two dimensions, even when servers are constrained to lie on a line (also known as the $1.5$D case), and for $\alpha=1$.
We show that in the instance illustrated in Figure~\ref{fig:15d}, left, Greedy is a factor of $2m/\sqrt{5}$ worse than the optimal solution.

In this instance, a set of $m$ servers are placed along a horizontal line with a distance of $1$ between each consecutive pair. Above each server, we place a ``column'' of clients stretching up to distance $m$ above the servers. The clients in a column are evenly spaced and at distance $d=\sqrt{m^2+1}-m$ of each other. Thus, there are $m/d = m(m+\sqrt{m^2+1})$ clients in each column. The total number of clients is roughly $2m^3$.

\begin{lemma}
Greedy's approximation ratio in the 1.5D setting with $\alpha=1$ is no better than $2m/\sqrt{5}$.
\end{lemma}

\begin{proof}
Consider the instance described above and illustrated in Figure~\ref{fig:15d}.

The optimal solution is to cover all clients with a single server located at the center. By the Pythagorean Theorem, the cost of the optimal solution is $\sqrt{5}m/2$. In contrast, we show that Greedy would choose to cover the clients in each column by the server at the bottom of it, resulting in a cost of $m^2$. Thus, Greedy is $2m/\sqrt{5}$ times worse than the optimal solution.

We assume that Greedy breaks ties by choosing clients closer to the horizontal line first (alternatively, we can perturb the positions of the clients slightly to guarantee this tie breaking). Then, we can show that Greedy covers the clients by ``layers'', where a layer is the set of clients at a given height. To see this, assume, for the sake of an inductive argument, that Greedy has grown each disk to cover the clients up to a given layer. Then, some server $s$ expands to cover a client $p$ in next layer, which would be the one in the same column. Let $s'$ and $p'$ be the server and client next to $s$ and $p$, respectively. We must argue that the disk of $s$ is not closer to $p'$ than the disk of $s'$. Note that it suffices to show that this does not happen when $p$ is the very last client in the column above $s$. This is because the further away $p$ is from $s$, the bigger the radius of the disk of $s$, resulting in a disk closer to $p'$. This is illustrated in Figure~\ref{fig:15d}, right. The distance $d$ between two consecutive points in a column is chosen precisely so that, in this scenario where $p$ is the last point, the disk of $s$ is exactly as close to $p'$ as the disk of $s'$. Depending on the tie-breaking rule (or changing $d$ to be slightly smaller), $s'$ will grow to cover $p'$ and not $s$.
\end{proof}

\begin{figure}
    \centering
    \includegraphics[width=0.8\textwidth]{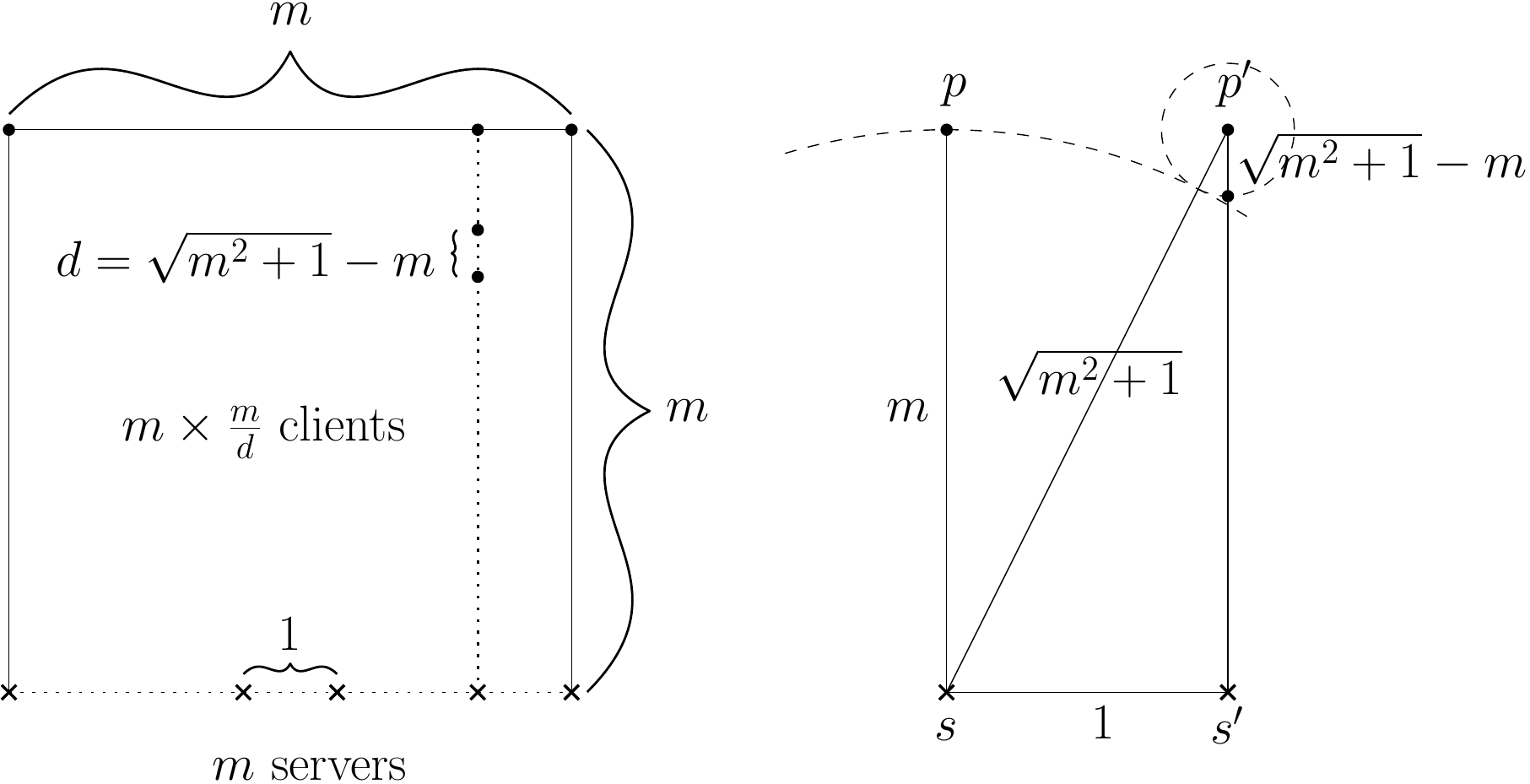}
    \caption{\textbf{Left:} Bad instance for the greedy algorithm for server cover. \textbf{Right:} Illustration (not to scale) that the disk of $s$ is not closer to $p'$ than the client below $p'$.}
    \label{fig:15d}
\end{figure}

%% file: main.bbl
\begin{thebibliography}{10}

\bibitem{pankaj99range}
Pankaj~K. Agarwal and Jeff Erickson.
\newblock Geometric range searching and its relatives.
\newblock In {\em Advances in discrete and computational geometry ({S}outh
  {H}adley, {MA}, 1996)}, volume 223 of {\em Contemp. Math.}, pages 1--56.
  Amer. Math. Soc., Providence, RI, 1999.
\newblock URL: \url{https://doi.org/10.1090/conm/223/03131}, \href
  {http://dx.doi.org/10.1090/conm/223/03131}
  {\path{doi:10.1090/conm/223/03131}}.

\bibitem{agarwal93}
Pankaj~K. Agarwal and Ji\v{r}\'i Matou\v{s}ek.
\newblock Ray shooting and parametric search.
\newblock {\em SIAM Journal on Computing}, 22(4):794--806, 1993.
\newblock URL: \url{https://doi.org/10.1137/0222051}, \href
  {http://arxiv.org/abs/https://doi.org/10.1137/0222051}
  {\path{arXiv:https://doi.org/10.1137/0222051}}, \href
  {http://dx.doi.org/10.1137/0222051} {\path{doi:10.1137/0222051}}.

\bibitem{agarwal2014near}
Pankaj~K. Agarwal and Jiangwei Pan.
\newblock Near-linear algorithms for geometric hitting sets and set covers.
\newblock In {\em Proceedings of the thirtieth annual symposium on
  Computational geometry}, page 271. ACM, 2014.

\bibitem{Aggarwal2009}
Gagan Aggarwal, S.~Muthukrishnan, D\'{a}vid P\'{a}l, and Martin P\'{a}l.
\newblock General auction mechanism for search advertising.
\newblock In {\em 18th Int. Conf. on the World Wide Web (WWW)}, pages 241--250.
  ACM, 2009.
\newblock \href {http://dx.doi.org/10.1145/1526709.1526742}
  {\path{doi:10.1145/1526709.1526742}}.

\bibitem{aichholzer1996novel}
Oswin Aichholzer, Franz Aurenhammer, David Alberts, and Bernd G{\"a}rtner.
\newblock A novel type of skeleton for polygons.
\newblock In {\em J. UCS The Journal of Universal Computer Science}, pages
  752--761. Springer, 1996.

\bibitem{carrots}
Helmut Alt, Esther~M. Arkin, Herv{\'e} Br{\"o}nnimann, Jeff Erickson,
  S{\'a}ndor~P. Fekete, Christian Knauer, Jonathan Lenchner, Joseph S.~B.
  Mitchell, and Kim Whittlesey.
\newblock Minimum-cost coverage of point sets by disks.
\newblock In {\em Proceedings of the twenty-second annual symposium on
  Computational geometry}, pages 449--458. ACM, 2006.

\bibitem{arkin2009geometric}
Esther~M Arkin, Sang~Won Bae, Alon Efrat, Kazuya Okamoto, Joseph~SB Mitchell,
  and Valentin Polishchuk.
\newblock Geometric stable roommates.
\newblock {\em Information Processing Letters}, 109(4):219--224, 2009.

\bibitem{arora1998polynomial}
Sanjeev Arora.
\newblock Polynomial time approximation schemes for euclidean traveling
  salesman and other geometric problems.
\newblock {\em Journal of the ACM (JACM)}, 45(5):753--782, 1998.

\bibitem{arya1998optimal}
Sunil Arya, David~M. Mount, Nathan~S. Netanyahu, Ruth Silverman, and Angela~Y.
  Wu.
\newblock An optimal algorithm for approximate nearest neighbor searching fixed
  dimensions.
\newblock {\em Journal of the ACM (JACM)}, 45(6):891--923, 1998.

\bibitem{avis83}
David Avis.
\newblock A survey of heuristics for the weighted matching problem.
\newblock {\em Networks}, 13(4):475--493, 1983.
\newblock URL:
  \url{https://onlinelibrary.wiley.com/doi/abs/10.1002/net.3230130404}, \href
  {http://arxiv.org/abs/https://onlinelibrary.wiley.com/doi/pdf/10.1002/net.3230130404}
  {\path{arXiv:https://onlinelibrary.wiley.com/doi/pdf/10.1002/net.3230130404}},
  \href {http://dx.doi.org/10.1002/net.3230130404}
  {\path{doi:10.1002/net.3230130404}}.

\bibitem{Barequet2003}
Gill Barequet, Michael~T. Goodrich, Aya Levi-Steiner, and Dvir Steiner.
\newblock Straight-skeleton based contour interpolation.
\newblock In {\em Proceedings of the Fourteenth Annual ACM-SIAM Symposium on
  Discrete Algorithms}, SODA '03, pages 119--127, Philadelphia, PA, USA, 2003.
  Society for Industrial and Applied Mathematics.
\newblock URL: \url{http://dl.acm.org/citation.cfm?id=644108.644129}.

\bibitem{bentley92}
Jon~Jouis Bentley.
\newblock Fast algorithms for geometric traveling salesman problems.
\newblock {\em ORSA Journal on Computing}, 4(4):387--411, 1992.
\newblock URL: \url{https://doi.org/10.1287/ijoc.4.4.387}, \href
  {http://arxiv.org/abs/https://doi.org/10.1287/ijoc.4.4.387}
  {\path{arXiv:https://doi.org/10.1287/ijoc.4.4.387}}, \href
  {http://dx.doi.org/10.1287/ijoc.4.4.387} {\path{doi:10.1287/ijoc.4.4.387}}.

\bibitem{Bentley1990}
Jon~Louis Bentley.
\newblock Experiments on traveling salesman heuristics.
\newblock In {\em Proceedings of the First Annual ACM-SIAM Symposium on
  Discrete Algorithms}, SODA '90, pages 91--99, Philadelphia, PA, USA, 1990.
  Society for Industrial and Applied Mathematics.
\newblock URL: \url{http://dl.acm.org/citation.cfm?id=320176.320186}.

\bibitem{Bespamyatnikh1998}
Sergei~N. Bespamyatnikh.
\newblock An optimal algorithm for closest-pair maintenance.
\newblock {\em Discrete {\&} Computational Geometry}, 19(2):175--195, Feb 1998.
\newblock URL: \url{https://doi.org/10.1007/PL00009340}, \href
  {http://dx.doi.org/10.1007/PL00009340} {\path{doi:10.1007/PL00009340}}.

\bibitem{bhatnagar2008}
Nayantara Bhatnagar, Sam Greenberg, and Dana Randall.
\newblock Sampling stable marriages: why spouse-swapping won't work.
\newblock In {\em Proceedings of the nineteenth annual ACM-SIAM symposium on
  Discrete algorithms}, pages 1223--1232. Society for Industrial and Applied
  Mathematics, 2008.

\bibitem{Bilo2005}
Vittorio Bil{\`{o}}, Ioannis Caragiannis, Christos Kaklamanis, and Panagiotis
  Kanellopoulos.
\newblock Geometric clustering to minimize the sum of cluster sizes.
\newblock In {\em Algorithms - {ESA} 2005, 13th Annual European Symposium,
  Palma de Mallorca, Spain, October 3-6, 2005, Proceedings}, pages 460--471,
  2005.
\newblock URL: \url{https://doi.org/10.1007/11561071\_42}, \href
  {http://dx.doi.org/10.1007/11561071\_42} {\path{doi:10.1007/11561071\_42}}.

\bibitem{BiniazApprox}
Ahmad Biniaz, Prosenjit Bose, Paz Carmi, Anil Maheshwari, J.~Ian Munro, and
  Michiel H.~M. Smid.
\newblock Faster algorithms for some optimization problems on collinear points.
\newblock {\em CoRR}, abs/1802.09505, 2018.
\newblock URL: \url{http://arxiv.org/abs/1802.09505}, \href
  {http://arxiv.org/abs/1802.09505} {\path{arXiv:1802.09505}}.

\bibitem{bowers2014faster}
John~C. Bowers.
\newblock Faster reductions for straight skeletons to motorcycle graphs, 2014.
\newblock \href {http://arxiv.org/abs/1405.6260} {\path{arXiv:1405.6260}}.

\bibitem{Brecklinghaus15}
Judith Brecklinghaus and Stefan Hougardy.
\newblock The approximation ratio of the greedy algorithm for the metric
  traveling salesman problem.
\newblock {\em Operations Research Letters}, 43(3):259 -- 261, 2015.
\newblock URL:
  \url{http://www.sciencedirect.com/science/article/pii/S0167637715000280},
  \href {http://dx.doi.org/https://doi.org/10.1016/j.orl.2015.02.009}
  {\path{doi:https://doi.org/10.1016/j.orl.2015.02.009}}.

\bibitem{bronnimann1995almost}
Herv{\'e} Br{\"o}nnimann and Michael~T. Goodrich.
\newblock Almost optimal set covers in finite vc-dimension.
\newblock {\em Discrete \& Computational Geometry}, 14(4):463--479, 1995.

\bibitem{Bruynooghe77}
Michel Bruynooghe.
\newblock New methods in automatic classification of numerous taxonomic data.
\newblock {\em Statistics and data analysis}, 2(3):24--42, 1977.
\newblock URL: \url{http://www.numdam.org/item/SAD_1977__2_3_24_0}.

\bibitem{bruynooghe1978}
Michel Bruynooghe.
\newblock Classification ascendante hi{\'e}rarchique des grands ensembles de
  donn{\'e}es: un algorithme rapide fond{\'e} sur la construction des
  voisinages r{\'e}ductibles.
\newblock {\em Les cahiers de l’analyse de donn{\'e}es}, 3:7--33, 1978.

\bibitem{cacciola04}
Fernando Cacciola.
\newblock A cgal implementation of the straight skeleton of a simple 2d polygon
  with holes.
\newblock {\em 2nd CGAL User Workshop}, 01 2004.
\newblock URL:
  \url{http://www.cgal.org/UserWorkshop/2004/straight\_skeleton.pdf}.

\bibitem{Chan2010}
Timothy~M. Chan.
\newblock A dynamic data structure for 3-d convex hulls and 2-d nearest
  neighbor queries.
\newblock In {\em Proceedings of the seventeenth annual ACM-SIAM symposium on
  Discrete algorithm}, pages 1196--1202. Society for Industrial and Applied
  Mathematics, 2006.

\bibitem{chennumber}
Jiehua Chen and Ugo~P. Finnendahl.
\newblock On the number of single-peaked narcissistic or single-crossing
  narcissistic preference profiles.
\newblock {\em Discrete Mathematics}, 341(5):1225 -- 1236, 2018.
\newblock URL:
  \url{http://www.sciencedirect.com/science/article/pii/S0012365X18300104},
  \href {http://dx.doi.org/https://doi.org/10.1016/j.disc.2018.01.008}
  {\path{doi:https://doi.org/10.1016/j.disc.2018.01.008}}.

\bibitem{Cheng2016}
Siu-Wing Cheng, Liam Mencel, and Antoine Vigneron.
\newblock A faster algorithm for computing straight skeletons.
\newblock {\em ACM Trans. Algorithms}, 12(3):44:1--44:21, April 2016.
\newblock URL: \url{http://doi.acm.org/10.1145/2898961}, \href
  {http://dx.doi.org/10.1145/2898961} {\path{doi:10.1145/2898961}}.

\bibitem{Cheng2007}
Siu-Wing Cheng and Antoine Vigneron.
\newblock Motorcycle graphs and straight skeletons.
\newblock {\em Algorithmica}, 47(2):159--182, Feb 2007.
\newblock URL: \url{https://doi.org/10.1007/s00453-006-1229-7}, \href
  {http://dx.doi.org/10.1007/s00453-006-1229-7}
  {\path{doi:10.1007/s00453-006-1229-7}}.

\bibitem{cloppet2000}
F.~Cloppet, J.~M. Oliva, and G.~Stamon.
\newblock Angular bisector network, a simplified generalized voronoi diagram:
  application to processing complex intersections in biomedical images.
\newblock {\em IEEE Transactions on Pattern Analysis and Machine Intelligence},
  22(1):120--128, Jan 2000.
\newblock \href {http://dx.doi.org/10.1109/34.824824}
  {\path{doi:10.1109/34.824824}}.

\bibitem{Cornuejols1985}
G{\'e}rard Cornu{\'e}jols, Jean Fonlupt, and Denis Naddef.
\newblock The traveling salesman problem on a graph and some related integer
  polyhedra.
\newblock {\em Mathematical Programming}, 33(1):1--27, Sep 1985.
\newblock URL: \url{https://doi.org/10.1007/BF01582008}, \href
  {http://dx.doi.org/10.1007/BF01582008} {\path{doi:10.1007/BF01582008}}.

\bibitem{DEKOSTER2007481}
René de~Koster, Tho Le-Duc, and Kees~Jan Roodbergen.
\newblock Design and control of warehouse order picking: A literature review.
\newblock {\em European Journal of Operational Research}, 182(2):481 -- 501,
  2007.
\newblock URL:
  \url{http://www.sciencedirect.com/science/article/pii/S0377221706006473},
  \href {http://dx.doi.org/https://doi.org/10.1016/j.ejor.2006.07.009}
  {\path{doi:https://doi.org/10.1016/j.ejor.2006.07.009}}.

\bibitem{DEMAINE20003}
Erik~D. Demaine, Martin~L. Demaine, and Joseph S.~B. Mitchell.
\newblock Folding flat silhouettes and wrapping polyhedral packages: New
  results in computational origami.
\newblock {\em Computational Geometry}, 16(1):3 -- 21, 2000.
\newblock URL:
  \url{http://www.sciencedirect.com/science/article/pii/S0925772199000565},
  \href {http://dx.doi.org/https://doi.org/10.1016/S0925-7721(99)00056-5}
  {\path{doi:https://doi.org/10.1016/S0925-7721(99)00056-5}}.

\bibitem{DujEppWoo-SJDM-17}
Vida Dujmovi{\'c}, David Eppstein, and David~R. Wood.
\newblock {Structure of graphs with locally restricted crossings}.
\newblock {\em SIAM J. Discrete Mathematics}, 31(2):805{--}824, 2017.
\newblock \href {http://dx.doi.org/10.1137/16M1062879}
  {\path{doi:10.1137/16M1062879}}.

\bibitem{dvorak2016}
Zden\v{e}k Dvo\v{r}{\'a}k and Sergey Norin.
\newblock Strongly sublinear separators and polynomial expansion.
\newblock {\em SIAM Journal on Discrete Mathematics}, 30(2):1095--1101, 2016.

\bibitem{Krari17}
Mehdi El~Krari, Bela{\"i}d Ahiod, and Bouazza El~Benani.
\newblock An empirical study of the multi-fragment tour construction algorithm
  for the travelling salesman problem.
\newblock In Ajith Abraham, Abdelkrim Haqiq, Adel~M. Alimi, Ghita Mezzour,
  Nizar Rokbani, and Azah~Kamilah Muda, editors, {\em Proceedings of the 16th
  International Conference on Hybrid Intelligent Systems (HIS 2016)}, pages
  278--287, Cham, 2017. Springer International Publishing.

\bibitem{Eppstein1999}
D.~Eppstein and J.~Erickson.
\newblock Raising roofs, crashing cycles, and playing pool: Applications of a
  data structure for finding pairwise interactions.
\newblock {\em Discrete {\&} Computational Geometry}, 22(4):569--592, Dec 1999.
\newblock URL: \url{https://doi.org/10.1007/PL00009479}, \href
  {http://dx.doi.org/10.1007/PL00009479} {\path{doi:10.1007/PL00009479}}.

\bibitem{eppstein2000fast}
David Eppstein.
\newblock Fast hierarchical clustering and other applications of dynamic
  closest pairs.
\newblock {\em Journal of Experimental Algorithmics (JEA)}, 5:1, 2000.

\bibitem{eppstein2017_2}
David Eppstein, Michael~T. Goodrich, Doruk Korkmaz, and Nil Mamano.
\newblock Defining equitable geographic districts in road networks via stable
  matching.
\newblock In {\em Proceedings of the 25th ACM SIGSPATIAL International
  Conference on Advances in Geographic Information Systems}, page~52. ACM,
  2017.

\bibitem{eppstein2017}
David Eppstein, Michael~T. Goodrich, and Nil Mamano.
\newblock Algorithms for stable matching and clustering in a grid.
\newblock In {\em International Workshop on Combinatorial Image Analysis},
  pages 117--131. Springer, 2017.

\bibitem{Eppstein17Latin}
David Eppstein, Michael~T. Goodrich, and Nil Mamano.
\newblock Reactive proximity data structures for graphs.
\newblock In Michael~A. Bender, Mart{\'i}n Farach-Colton, and Miguel~A.
  Mosteiro, editors, {\em LATIN 2018: Theoretical Informatics}, pages 777--789,
  Cham, 2018. Springer International Publishing.

\bibitem{eppstein2017crossing}
David Eppstein and Siddharth Gupta.
\newblock Crossing patterns in nonplanar road networks.
\newblock In {\em 25th ACM SIGSPATIAL Int. Conf. on Advances in Geographic
  Information Systems}, 09 2017.

\bibitem{gale62}
David Gale and Lloyd~S. Shapley.
\newblock {College admissions and the stability of marriage}.
\newblock {\em The American Mathematical Monthly}, 69(1):9{--}15, 1962.
\newblock \href {http://dx.doi.org/10.2307/2312726}
  {\path{doi:10.2307/2312726}}.

\bibitem{GibsonClustering}
Matt Gibson, Gaurav Kanade, Erik Krohn, Imran~A. Pirwani, and Kasturi
  Varadarajan.
\newblock On clustering to minimize the sum of radii.
\newblock {\em SIAM J. Comput.}, 41(1):47--60, January 2012.
\newblock URL: \url{http://dx.doi.org/10.1137/100798144}, \href
  {http://dx.doi.org/10.1137/100798144} {\path{doi:10.1137/100798144}}.

\bibitem{gilbert1984}
John~R. Gilbert, Joan~P. Hutchinson, and Robert~E. Tarjan.
\newblock A separator theorem for graphs of bounded genus.
\newblock {\em Journal of Algorithms}, 5(3):391--407, 1984.

\bibitem{Gonczarowski2015}
Yannai~A. Gonczarowski, Noam Nisan, Rafail Ostrovsky, and Will Rosenbaum.
\newblock A stable marriage requires communication.
\newblock In {\em Proceedings of the Twenty-sixth Annual ACM-SIAM Symposium on
  Discrete Algorithms}, SODA '15, pages 1003--1017, Philadelphia, PA, USA,
  2015. Society for Industrial and Applied Mathematics.
\newblock URL: \url{http://dl.acm.org/citation.cfm?id=2722129.2722197}.

\bibitem{Goodrich1993}
Michael~T. Goodrich and Roberto Tamassia.
\newblock Dynamic ray shooting and shortest paths via balanced geodesic
  triangulations.
\newblock In {\em Proceedings of the Ninth Annual Symposium on Computational
  Geometry}, SCG '93, pages 318--327, New York, NY, USA, 1993. ACM.
\newblock URL: \url{http://doi.acm.org/10.1145/160985.161157}, \href
  {http://dx.doi.org/10.1145/160985.161157} {\path{doi:10.1145/160985.161157}}.

\bibitem{HENZINGER19973}
Monika~R. Henzinger, Philip Klein, Satish Rao, and Sairam Subramanian.
\newblock {Faster shortest-path algorithms for planar graphs}.
\newblock {\em Journal of Computer and System Sciences}, 55(1):3{--}23, 1997.
\newblock \href {http://dx.doi.org/10.1006/jcss.1997.1493}
  {\path{doi:10.1006/jcss.1997.1493}}.

\bibitem{Hershberger1992}
John Hershberger and Subhash Suri.
\newblock Applications of a semi-dynamic convex hull algorithm.
\newblock {\em BIT Numerical Mathematics}, 32(2):249--267, Jun 1992.
\newblock URL: \url{https://doi.org/10.1007/BF01994880}, \href
  {http://dx.doi.org/10.1007/BF01994880} {\path{doi:10.1007/BF01994880}}.

\bibitem{Hoepman2004SimpleDW}
Jaap-Henk Hoepman.
\newblock Simple distributed weighted matchings.
\newblock {\em CoRR}, cs.DC/0410047, 2004.

\bibitem{Huber2011}
Stefan Huber and Martin Held.
\newblock Theoretical and practical results on straight skeletons of planar
  straight-line graphs.
\newblock In {\em Proceedings of the Twenty-seventh Annual Symposium on
  Computational Geometry}, SoCG '11, pages 171--178, New York, NY, USA, 2011.
  ACM.
\newblock URL: \url{http://doi.acm.org/10.1145/1998196.1998223}, \href
  {http://dx.doi.org/10.1145/1998196.1998223}
  {\path{doi:10.1145/1998196.1998223}}.

\bibitem{huber12}
Stefan Huber and Martin Held.
\newblock A fast straight-skeleton algorithm based on generalized motorcycle
  graphs.
\newblock {\em International Journal of Computational Geometry \&
  Applications}, 22(05):471--498, 2012.
\newblock URL: \url{https://doi.org/10.1142/S0218195912500124}, \href
  {http://arxiv.org/abs/https://doi.org/10.1142/S0218195912500124}
  {\path{arXiv:https://doi.org/10.1142/S0218195912500124}}, \href
  {http://dx.doi.org/10.1142/S0218195912500124}
  {\path{doi:10.1142/S0218195912500124}}.

\bibitem{JohnMcGe97}
David~S. Johnson and Lyle~A. McGeoch.
\newblock The traveling salesman problem: A case study in local optimization.
\newblock In E.~H.~L. Aarts and J.~K. Lenstra, editors, {\em Local Search in
  Combinatorial Optimization}, pages 215--310. John Wiley and Sons, Chichester,
  United Kingdom, 1997.

\bibitem{KapMulRod-16}
Haim Kaplan, Wolfgang Mulzer, Liam Roditty, Paul Seiferth, and Micha Sharir.
\newblock {Dynamic planar Voronoi diagrams for general distance functions and
  their algorithmic applications}.
\newblock In {\em 28th ACM-SIAM Symp. on Discrete Algorithms (SODA)}, pages
  2495{--}2504, 2017.
\newblock \href {http://dx.doi.org/10.1137/1.9781611974782.165}
  {\path{doi:10.1137/1.9781611974782.165}}.

\bibitem{kawarabayashi2010}
Ken-ichi Kawarabayashi and Bruce Reed.
\newblock A separator theorem in minor-closed classes.
\newblock In {\em 51st IEEE Symp. on Foundations of Computer Science (FOCS)},
  pages 153--162, 2010.

\bibitem{LevTov05}
Nissan Lev-Tov and David Peleg.
\newblock Polynomial time approximation schemes for base station coverage with
  minimum total radii.
\newblock {\em Comput. Netw.}, 47(4):489--501, March 2005.
\newblock URL: \url{http://dx.doi.org/10.1016/j.comnet.2004.08.012}, \href
  {http://dx.doi.org/10.1016/j.comnet.2004.08.012}
  {\path{doi:10.1016/j.comnet.2004.08.012}}.

\bibitem{Mato1992trees}
Ji{\v{r}}{\'i} Matou{\v{s}}ek.
\newblock Efficient partition trees.
\newblock {\em Discrete {\&} Computational Geometry}, 8(3):315--334, Sep 1992.
\newblock URL: \url{https://doi.org/10.1007/BF02293051}, \href
  {http://dx.doi.org/10.1007/BF02293051} {\path{doi:10.1007/BF02293051}}.

\bibitem{Matousek1992}
Ji\v{r}\'{\i} Matou\v{s}ek and Otfried Schwarzkopf.
\newblock Linear optimization queries.
\newblock In {\em Proceedings of the Eighth Annual Symposium on Computational
  Geometry}, SCG '92, pages 16--25, New York, NY, USA, 1992. ACM.
\newblock URL: \url{http://doi.acm.org/10.1145/142675.142683}, \href
  {http://dx.doi.org/10.1145/142675.142683} {\path{doi:10.1145/142675.142683}}.

\bibitem{misev11}
Alfonsas Misevicius and Andrius Blazinskas.
\newblock Combining 2-opt, 3-opt and 4-opt with k-swap-kick perturbations for
  the traveling salesman problem.
\newblock {\em 17th International Conference on Information and Software
  Technologies}, 2011.

\bibitem{moeller2016}
Daniel Moeller, Ramamohan Paturi, and Stefan Schneider.
\newblock Subquadratic algorithms for succinct stable matching.
\newblock In {\em International Computer Science Symposium in Russia}, pages
  294--308. Springer, 2016.

\bibitem{Moscato1994AnAO}
Pablo Moscato and Michael~G. Norman.
\newblock On the performance of heuristics on finite and infinite fractal
  instances of the euclidean traveling salesman problem.
\newblock {\em {INFORMS} Journal on Computing}, 10(2):121--132, 1998.
\newblock URL: \url{https://doi.org/10.1287/ijoc.10.2.121}, \href
  {http://dx.doi.org/10.1287/ijoc.10.2.121} {\path{doi:10.1287/ijoc.10.2.121}}.

\bibitem{Muellner2011}
Daniel {M{\"u}llner}.
\newblock {Modern hierarchical, agglomerative clustering algorithms}.
\newblock {\em arXiv e-prints}, September 2011.
\newblock \href {http://arxiv.org/abs/1109.2378} {\path{arXiv:1109.2378}}.

\bibitem{murtagh1983}
Fionn Murtagh.
\newblock A survey of recent advances in hierarchical clustering algorithms.
\newblock {\em The Computer Journal}, 26(4):354--359, 1983.

\bibitem{thematch}
{National Resident Matching Program}, 2017.
\newblock URL: \url{http://www.nrmp.org}.

\bibitem{oliva1996}
J.~M. Oliva, M.~Perrin, and S.~Coquillart.
\newblock 3d reconstruction of complex polyhedral shapes from contours using a
  simplified generalized voronoi diagram.
\newblock {\em Computer Graphics Forum}, 15(3):397--408, 1996.
\newblock URL:
  \url{https://onlinelibrary.wiley.com/doi/abs/10.1111/1467-8659.1530397},
  \href
  {http://arxiv.org/abs/https://onlinelibrary.wiley.com/doi/pdf/10.1111/1467-8659.1530397}
  {\path{arXiv:https://onlinelibrary.wiley.com/doi/pdf/10.1111/1467-8659.1530397}},
  \href {http://dx.doi.org/10.1111/1467-8659.1530397}
  {\path{doi:10.1111/1467-8659.1530397}}.

\bibitem{ONG1984273}
Hoon~Liong Ong and J.~B. Moore.
\newblock Worst-case analysis of two travelling salesman heuristics.
\newblock {\em Operations Research Letters}, 2(6):273 -- 277, 1984.
\newblock URL:
  \url{http://www.sciencedirect.com/science/article/pii/0167637784900786},
  \href {http://dx.doi.org/https://doi.org/10.1016/0167-6377(84)90078-6}
  {\path{doi:https://doi.org/10.1016/0167-6377(84)90078-6}}.

\bibitem{PedersenWang18}
Logan Pedersen and Haitao Wang.
\newblock On the coverage of points in the plane by disks centered at a line.
\newblock In {\em Proceedings of the 30th Canadian Conference on Computational
  Geometry, {CCCG} 2018, August 8-10, 2018, University of Manitoba, Winnipeg,
  Manitoba, Canada}, pages 158--164, 2018.
\newblock URL:
  \url{http://www.cs.umanitoba.ca/\%7Ecccg2018/papers/session4A-p1.pdf}.

\bibitem{shapley1974cores}
Lloyd Shapley and Herbert Scarf.
\newblock On cores and indivisibility.
\newblock {\em Journal of mathematical economics}, 1(1):23--37, 1974.

\bibitem{Vigneron2014}
Antoine Vigneron and Lie Yan.
\newblock A faster algorithm for computing motorcycle graphs.
\newblock {\em Discrete Comput. Geom.}, 52(3):492--514, October 2014.
\newblock URL: \url{http://dx.doi.org/10.1007/s00454-014-9625-2}, \href
  {http://dx.doi.org/10.1007/s00454-014-9625-2}
  {\path{doi:10.1007/s00454-014-9625-2}}.

\bibitem{ZHANG201530}
Huili Zhang, Weitian Tong, Yinfeng Xu, and Guohui Lin.
\newblock The steiner traveling salesman problem with online edge blockages.
\newblock {\em European Journal of Operational Research}, 243(1):30 -- 40,
  2015.
\newblock URL:
  \url{http://www.sciencedirect.com/science/article/pii/S0377221714009175},
  \href {http://dx.doi.org/https://doi.org/10.1016/j.ejor.2014.11.013}
  {\path{doi:https://doi.org/10.1016/j.ejor.2014.11.013}}.

\end{thebibliography}
